\documentclass[10pt, article, twocolumn]{IEEEtran}
\usepackage[pdftex]{graphicx}
 \pdfoutput=1
\usepackage{enumerate}
\usepackage{caption}
\usepackage{subcaption}
\usepackage{amsthm}
\usepackage{amssymb}
\usepackage{dsfont}
\usepackage{setspace}
\usepackage[linesnumbered,ruled,vlined]{algorithm2e}
\usepackage{comment}
\usepackage{color}
\usepackage{float}
\usepackage{gensymb}
\usepackage{cuted}
\usepackage{soul}
\theoremstyle{plain}
\newtheorem{corollary}{Corollary}

\newtheorem{theorem}{Theorem}
\newtheorem{lemma}{Lemma}
\newtheorem{fact}{Fact}
\newtheorem{definition}{Definition}
\newtheorem{assumption}{Assumption}

\theoremstyle{definition}

\theoremstyle{remark}
\newtheorem{remark}{Remark}

\usepackage{mathtools}
\newcommand{\expect}{\operatorname{\mathbb{E}}\expectarg}
\DeclarePairedDelimiterX{\expectarg}[1]{[}{]}{%
  \ifnum\currentgrouptype=16 \else\begingroup\fi
  \activatebar#1
  \ifnum\currentgrouptype=16 \else\endgroup\fi
}
\newcommand{\prob}{\operatorname{\mathbb{P}}\probarg}
\DeclarePairedDelimiterX{\probarg}[1]{(}{)}{%
  \ifnum\currentgrouptype=16 \else\begingroup\fi
  \activatebar#1
  \ifnum\currentgrouptype=16 \else\endgroup\fi
}
\newcommand{\innermid}{\nonscript\;\delimsize\vert\nonscript\;}
\newcommand{\activatebar}{%
  \begingroup\lccode`\~=`\|
  \lowercase{\endgroup\let~}\innermid 
  \mathcode`|=\string"8000
}

\def\*#1{\mathbf{#1}}
\def\&#1{\mathcal{#1}}
\def\.#1{\boldsymbol#1}
\def\^#1{\hat{#1}}
\def\[#1{\left #1}
\def\]#1{\right #1}

\DeclareMathOperator*{\argmin}{arg\,min\ }
\DeclareMathOperator*{\argmax}{arg\,max\ }

\DeclareMathAlphabet\mathbfcal{OMS}{cmsy}{b}{n}

\author{\IEEEauthorblockN{Sung-En Chiu, Nancy Ronquillo and Tara Javidi}
\vspace{2mm}

\IEEEauthorblockA{Department of Electrical and Computer Engineering \\
University of California, San Diego \\
Email: \{suchiu,nronquil,tjavidi\}@ucsd.edu}
}

\begin{document}
\singlespacing

\title{Active Learning and CSI Acquisition for mmWave Initial Alignment}
%\author{\IEEEauthorblockN{Sung-En Chiu and Bhaskar Rao}
%\vspace{2mm}
%\IEEEauthorblockA{Department of Electrical and Computer Engineering \\
%University of California, San Diego \\
%%Email: \{suchiu,brao\}@ucsd.edu}
%}

%This paper considers the problem of adaptive and sequential optimization of the beamforming vectors during the initial access phase of communication. With a single-path channel model, the problem is reduced to actively learning the Angle-of-Arrival (AoA) of the signal sent from the user to the Base Station (BS). Drawing on the recent results in the design of a hierarchical beamforming (Alkhateeb'14 et.al.), active learning from an imperfect labeler (Yan'16 et.al.), and measurement dependent noisy search strategies (Chiu'16 et.al.), an adaptive and sequential alignment algorithm is proposed. For any given resolution and error probability of the estimated AoA, an upper bound on the expected time of the initial alignment of the proposed algorithm is derived via the Extrinsic Jensen Shannon Divergence. Numerically, the proposed algorithm is compared with prior works. A significant improvement of the system communication rate is observed, most notably in the regime when SNR is low, making our algorithm the most suitable for stand-alone mmWave communication. 

\maketitle

\begin{abstract}
Millimeter wave (mmWave) communication with large antenna arrays is a promising technique to enable extremely high data rates due to large available bandwidth in mmWave frequency bands. {In addition,} given the knowledge of an optimal directional beamforming vector, large antenna arrays have been shown to overcome both the severe signal attenuation in mmWave as well as the interference problem. However, fundamental limits on  achievable learning {rate}  of an optimal beamforming vector remain.

This paper considers the problem of adaptive and sequential optimization of the beamforming vectors during the initial access phase of communication. With a single-path channel model, the problem is reduced to actively learning the Angle-of-Arrival (AoA) of the signal sent from the user to the Base Station (BS). Drawing on the recent results in the design of a hierarchical beamforming codebook \cite{Alkhateeb2014}, sequential measurement dependent noisy search strategies \cite{Chiu2016}, and active learning from an imperfect labeler \cite{Yan2016}, an adaptive and sequential alignment algorithm is proposed.
%over the hierarchical beamforming codebook proposed by Alkhateeb, Heath, et. al. 
%A hierarchical beamforming codebook is used for searching the Angle-of-Arrival (AoA) of the signal path. By sequentially choosing the beamforming vectors from the codebook, a posterior-based noisy bisection search algorithm is proposed for quickly and reliably zooming in to the direction of AoA. 

For any given resolution and error probability of the estimated AoA, an upper bound on the expected search time of the proposed algorithm is derived via Extrinsic Jensen-Shannon Divergence. The upper bound demonstrates that the search time of the proposed algorithm asymptotically matches the performance of the noiseless bisection search up to a constant factor, {in effect,}  characterizing the AoA acquisition rate. Furthermore, the upper bound shows that the acquired AoA error probability decays exponentially fast with the search time with an exponent that is a decreasing function of the acquisition rate. 

Numerically, the proposed algorithm is compared with prior work where a significant improvement of the system communication rate is observed. Most notably, in the relevant regime of low ($-10$dB to $+5$dB) raw SNR, this establishes the first practically viable solution for initial access and, hence, the first demonstration of stand-alone mmWave communication. 

\end{abstract}

\section{Introduction}

%In this paper, we consider a beamforming design problem for millimeter wave (mmWave) communication. 

Millimeter wave (mmWave) communication with massive antenna arrays is a promising technique that increases the data rate significantly, thanks to the large available bandwidth in mmWave frequency bands. 
%One of the major challenges of using mmWave is that the pathloss in such high frequency bands is much higher \cite{Maccartney2013,Rappaport2017}, resulting in low signal-to-noise ratio (SNR) and link outage. To overcome this high pathloss, large antenna arrays with directional beamforming have been considered, which require the knowledge of the channel state information (CSI) between the transmit antennas and the received antennas. How to efficiently and accurately obtain the CSI and form the beam toward the right directions is thus the core problem for enabling mmWave communication.
While an inherent challenge for mmWave communication is extremely high pathloss \cite{Maccartney2013}-~\cite{Rappaport2017}, resulting in low SNR and high link outage, the small wavelength can be exploited to deploy an array with a very large number of antennas in a relatively small area. It has been shown \cite{Molisch2017} that massive MIMO mmWave systems can be deployed to form highly directional beams to mitigate the pathloss and the associated low SNR and high link outage. However, it is important to note that the realization of highly directional beams requires a precise and reliable estimate of channel state information (CSI) \cite{Nayebi17} during the initial access phase. This paper considers the problem of actively learning an optimal beamforming vector from a fundamental limit point of view.

%TO MOVE: (*)In fact, without sufficiently accurate and speedy acquisition of CSI, mmWave solutions reduce to non-stand alone networks in which the initial access process is left to conventional communication systems operating over lower frequencies. 

With the scale of millimeter wavelength and the half wavelength spacing, a large number of antennas can be packed into a modest-sized device. For large antenna arrays, however, equipping each antenna with an RF chain is too hardware costly. This prevents per antenna digital processing. A hardware friendly proposal for practically implementing large array systems in mmWave bands deploys a single RF chain where CSI acquisition reduces to finding the optimal analog beamforming along the dominant direction of the signals between the base station (BS) and the user that is trying to establish the communication link. In this paper we consider this practical scenario of mmWave communication with massive MIMO technology and the practically designed low complexity hierarchical beamforming codebook of \cite{Alkhateeb2014} to propose an efficient and adaptive beamforming algorithm that quickly identifies the optimal beamforming direction under a single dominant path channel model. Furthermore, we obtain bounds on the performance of this algorithm to asymptotically match the fundamental information theoretic limits on the speed and reliability of active learning and CSI acquisition with the given hardware constraints.

%Specifically, hybrid beamforming is separated into two stages, digital beamforming and analog beamforming, where the analog beamforming combines signals from a large number of antenna elements to a limited number of RF chains, creating several virtual antennas. Among these RF chains (virtual antennas), conventional digital beamforming can be applied. See \cite{Molisch2017} for a survey of hybrid beamforming. Conventional CSI acquisition can only be applied on the virtual antenna channel for digital beamforming with the support of the processing units on each of the RF chains. The acquisition of the CSI for the design of the analog beamforming arises as a new challenge. In particular, the direction of the signals is essential for the beam alignment especially in the initial access phase where the base station (BS) and the user is trying to establish the communication link. An 
The exhaustive linear search, which utilizes beams that scan over all possible directions to pick the best one, and is proposed in IEEE 802.15 and 5G standards, requires a relatively long initial access time that linearly grows with the angle resolution (highest resolution being the number of the antenna elements).
On the theoretical front, in contrast, prior work which is based on simple measurement models noted that the problem of CSI acquisition in mmWave is closely related to that of noisy search, which itself has been shown to be closely related to the problem of channel coding over a binary input channel with (\cite{Giordani2016,Kaspi2018,Lalitha2017}) and without (\cite{Shabara2017}) feedback. Under various noise models, it is shown that the number of measurements can be kept to grow only logarithmically with the angular resolution and target error probability \cite{Kaspi2018} and \cite{Lalitha2017}. 
While these early studies did not take the practical beam patterns into account, this logarithmic scaling was later also confirmed and reported in more practical systems with realistic and empirically precise beam patterns (\cite{Alkhateeb2014, Abari2016, Song19}) with the caveat of a sufficiently large SNR model. In particular, \cite{Alkhateeb2014} carefully developed a hierarchical beamforming codebook which in the noiseless setting allows for an (adaptive) binary search over the angular space; increasing transmission power and/or time is then used to combat the measurement noise. The authors in \cite{Abari2016,Song19} showed that similar performance gains can also be achieved by a non-adaptive strategy. More specifically, the authors of \cite{Abari2016} proposed random hash functions to generate a random beamforming codebook whose acquisition time, they showed, grows only logarithmically with target resolution/error probability. The logarithmic scaling (of search time with angular resolution) could also be obtained when viewing the problem as that of sparse estimation with compressive measurements (see \cite{Ding2018} and references therein). Indeed, the authors of \cite{Song19} recover the signal direction with a non-negative least square estimate from Compressive Sensing by measuring the received power via a random beamforming codebook which hashes the angular directions similarly to \cite{Abari2016}. 

However, to guarantee coverage in low ($<5$ dB)\footnote{We note that while \cite{Song19} provides a system working under $-30$ dB SNR$_{\text{BBF}}$ (before beamforming), the system parameters used in their simulation were such that the SNR$_{\text{BBF}}$ is defined differently than raw SNR in our set up. For a fair comparison we may interpret our raw SNR = $\frac{\text{SNR}_{\text{BBF}}}{c}$ where the scaling factor $c$ is proportional to number of subcarriers and number of RF chains used in their simulation.} raw SNR regimes (cell-edge users), these beamforming techniques (random direction and direct bisection) provide marginal advantage over the exhaustive linear search as noted in \cite{Giordani2016}. This limitation of prior work to operate in high raw SNR makes them unsuitable for cell-edge users in mmWave communication. This has major practical system design implications, namely the current 5G mmWave communication in 3GPP standards \cite{3gppTR37340} supports mainly non-standalone mmWave in which the initial access phase is covered by legacy sub-6G infra-structure which provides higher SNR. This highlights the need for a strategy that can adaptively improve the measurement quality and is suitable for a low raw SNR regime.

\subsection{Our work and contributions}
\begin{figure}
    \centering
    \includegraphics[width = 0.35 \textwidth]{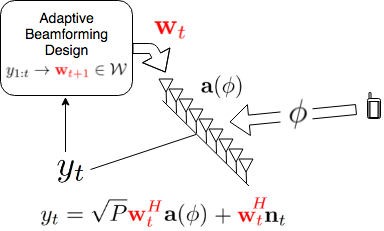}
    \caption{The active learning process of the AoA $\phi$ is to design the beams $\*w_t\in \mathcal{W}^S$ adaptively for the sequential collection of the observations $y_t$, from which at the ending of the collection is to be used for the estimation of the AoA $\phi$. }
    \label{fig:adaptive_beamforming}
\end{figure}

In this paper, we consider the problem of CSI acquisition during the initial access phase for designing the analog beamforming in an environment with a single-path channel. We formulate the CSI acquisition as an active learning of the angle-of-arrival (AoA) at the base station (BS) side where the user's beamforming vector is assumed to be fixed, as illustrated in Fig \ref{fig:adaptive_beamforming}. We consider two measures of performance for the proposed search scheme: the (expected) resolved beam width (AoA resolution) and the (expected) error probability. Based on the nature of the initial access protocol, we consider two types of stopping criteria: fixed-length stopping, where a fixed amount of search time is allocated for the initial access phase, and variable-length stopping, where search is conducted over a random stopping time. The contributions of the paper include:
\begin{itemize}
    \item We formulate the initial beam alignment for massive MIMO as active learning of the AoA through multiple sequential and adaptive search beams. Our approach draws heavily from our prior work on algorithms for noisy search \cite{Chiu2016}, active learning \cite{Naghshvar2015a}, and channel coding with feedback \cite{Lalitha2017}.

    \item We propose a new adaptive beamforming strategy that utilizes the hierarchical beamforming codebook of \cite{Alkhateeb2014}. The proposed adaptive strategy, hierarchical Posterior Matching ($hiePM$), accounts for the measurement noise and selects the beamforming vectors from the hierarchical beamforming codebook based on the posterior of the AoA. The design and analysis of $hiePM$ extends our prior work of sorted posterior matching for noisy search \cite{Chiu2016} and \cite{Lalitha2017} in that it restricts the search and the measurements to the practical and hierarchical beamforming patterns of \cite{Alkhateeb2014}.
    
    \item We analyze the proposed $hiePM$ strategy and give an upper bound on the expected acquisition time of a variable-length $hiePM$ search strategy required to reach a fixed (predetermined) target resolution and error probability in the AoA estimate. Even when the measurements are hard detected (1-bit quantized measurements), the achievable AoA acquisition rate and the error exponent of $hiePM$ are shown (Corollary~\ref{cor}) to be significantly better than those for the search methods of \cite{Alkhateeb2014} and the random hashing of \cite{Abari2016} in all raw SNR regimes.
    
    \item We show, via extensive simulations, that $hiePM$ is suitable for both fixed-length and variable-length initial access and significantly outperforms the state-of-art search strategies of \cite{Alkhateeb2014} and \cite{Abari2016}. The numerical simulations show that $hiePM$ is capable of reaching a good resolution and error probability even in a low ($<5$dB) raw SNR regime with reasonable expected search time overhead, demonstrating the possibility of stand-alone mmWave communication for the first time
    
\end{itemize}

\subsection{Notations} We use boldface letters to represent vectors and use $[n]$ as shorthand for the discrete set $\{1, 2, \ldots, n \}$. We denote the space of probability mass functions on set $\mathcal{X}$ as $P(x)$. We denote the Kullback-Leibler (KL) divergence between distribution $P$ and $Q$ by $D(P \| Q) = \sum_x P(x) \log \frac{P(x)}{Q(x)}$. The mutual information between random variable $X$ and $Y$ is defined as $I(X,Y) = \sum_{x,y} p(x,y) \log \frac{p(x,y)}{p(x)p(y)}$, where $p(x,y)$ is the joint distribution and $p(x)$ and $p(y)$ are the marginals of $X$ and $Y$. Let Bern$(p)$ denote the Bernoulli distribution with parameter $p$, and Bern$(x;p) = p^x (1-p)^{1-x}$. Let $I(q;p)$ denote the mutual information of the input $X \sim$ Bern$(q)$ and the output $Y$ of a BSC channel with crossover probability $p$. $C_1(p) := D( \text{Bern}(p) \| \text{Bern}(1-p) )$ denotes the error exponent of hypothesis testing of $\text{Bern}(p)$ versus $\text{Bern}(1-p)$. $\mathcal{CN}(\.\mu, \.\Sigma)$ denotes multivariate complex Gaussian distribution and $\mathcal{CN}(\*x ;\.\mu, \.\Sigma)$ with mean $\mu$ and covariance matrix $\.\Sigma$. $\text{Rice} (\mu,\sigma^2)$ denotes and Rician distribution and $\text{Rice} (x;\mu,\sigma^2):= \frac{x}{\sigma^2} \exp\[( \frac{-(x^2+\mu^2)}{2\sigma^2}\]) J_0(\frac{x \mu}{\sigma^2}) $ denotes its probability density function, where $J_0(\cdot)$ is the modified Bessel function of the first kind with order zero. 

\section{System Model}

We consider a sectorized cellular communication system operating in EHF (30-300 GHz) bands, where a sector is formed by the BS serving users in a range of angles {$[\theta_{\text{min}},\theta_{\text{max}}]$} as depicted in Fig.~\ref{fig:sector}.
\begin{figure}
    \centering
    \includegraphics[width = 0.4 \textwidth]{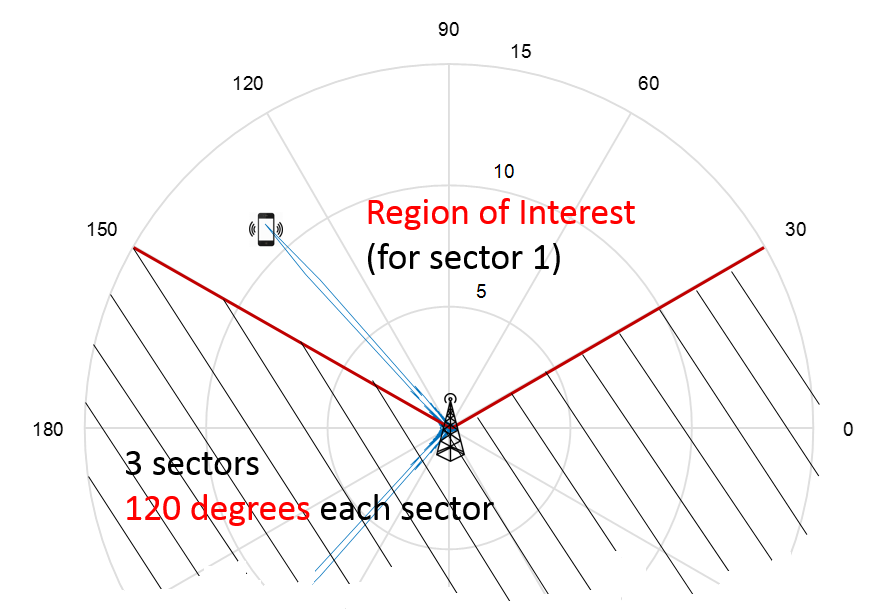}
    \caption{Base Station Serving sector $[30^o,150^o]$ and a received beam at the BS formed by uniform linear array}
    \label{fig:sector}
\end{figure}
We focus on the model with one sector and a single user, where the interference from other sectors are assumed to be negligible. This assumption is justified due to high pathloss in the EHF bands \cite{Maccartney2013}, and the orthogonality (in time or code)  of the transmissions from other users within the sector.

%\subsection{Transmit-Receive Signal Modelling with Multiple Antennas}
\label{System}

We consider a hardware architecture consisting of multiple antenna elements with a single RF chain \cite{Roh2014} on both the BS and the user sides. Beamforming is applied on the antenna elements such that the power gain due to beamforming may compensate the high pathloss in the mmWave communication system. We use a pilot-based procedure where the users send out pilots to the BS while the BS combines the signal from the antenna elements to the RF chain by the beamforming vector $\*w_t\in \mathbb{C}^N$. We will focus on the procedure of obtaining a good beamforming vector at the BS, while assuming a fixed beamforming vector at the user which allows us to model the user's antennas as a single virtual antenna. 

Let $N$ be the number of antennas at the BS, $\sqrt{P}$ be the combined effect of the transmit power and the large-scale fading (pathloss and shadowing), and $\*h \in \mathbb{C}^N$ be the small-scale frequency flat fading vector, i.e. $h_i$ is the small-scale fading between the single virtual antenna of the user and the $i$th antenna element at the BS. For small-scale channel modelling, we use the stochastic multi-path modelling (see Ch.7 in \cite{Tse2005}) assumption with a single dominant path. Furthermore, we assume that the user's mobility is negligible, i.e., the channel vector $\*h$ is time invariant. In summary, we have the following assumption:
\begin{assumption} \label{assum:single_path}
The small-scale channel can be described as:
\begin{equation} \label{eq:singlepath}
    \*h = \alpha \*a(\phi),
\end{equation}
where $\alpha \in \mathbb{C}$ is the fading coefficient and 
\begin{equation}
    \*a(\phi) :=  [ 1, e^{j\frac{2\pi d }{\lambda} \sin{\phi} },..., e^{j(N-1)\frac{2\pi d }{\lambda} \sin{\phi} }  ]
\end{equation}
\label{singlepath_alpha}
is the array manifold created by the Angle-of-Arrival (AoA) $\phi\in[\theta_{\text{min}},\theta_{\text{max}}]$ with antenna spacing $d$. Furthermore, we assume that the fading coefficient, $\alpha$, and the AoA, $\phi$, are static in time. 
\end{assumption}
%The digital signal processing unit is only available \textit{after} the RF chain processing (down-conversion, ADC, etc). Due to this hardware constraint of single RF chain, the signal at each antenna element is not available for complicated signal analysis, $e.g.$ conventional algorithms like MUSIC or MVDR can not applied. Therefore, even during the initial access phase, the BS has to first combine the signal from the antenna element to the RF chain by certain received beamforming vector $\*w_t\in \mathbb{C}^N$. 
Let time index $t=1,2,...$ be the time frame in which the BS can adapt the beamforming vector $\*w_t$. Each beamforming slot consists of $I$ samples of finer granularity either of time (e.g. CDMA) or of frequencies (e.g. OFDM subcarriers). Orthogonal spreading sequences $\mathbf{s}_k$ of length $I$ are sent by each of the $K$ users. In other words, we assume:
\begin{assumption} \label{assum:cdma}
\begin{equation}
     \mathbf{s}_k^H \mathbf{s}_{k'} = \begin{cases} 0 \text{ for } k = k'\\
1 \text{ for }k \neq k'\end{cases}
\end{equation}
\end{assumption}

With the assumption of orthogonality among users, correlating the pilot codes we can write the code-matched signal from a particular user as
\begin{equation}\label{eq:obsv}
\begin{aligned}
    y_t &{=}  \sqrt{P} \*w_t^H  (\sum_{k'=1}^K \*h_{k'}  \mathbf{s}_{k'}^{T} ) \mathbf{s}_k^{*} + \*w_t^H \*N_{t} \mathbf{s}_k^*\\
        &\stackrel{(a)}{=}  \alpha \sqrt{P} \*w_t^H \*a(\phi) + \*w_t^H \*n_{t},
\end{aligned}
\end{equation}
where $\*N_{t}$ is the $N \times I$  spatially uncorrelated AWGN noise matrix across the antenna elements (rows) and samples (columns). Note that in (a) we used single-path channel model (Assumption~\ref{assum:single_path}) and orthogonality of codes (Assumption~\ref{assum:cdma}) from different users as well as the static nature of the channel, $\*h$, over the code resource $I$. Finally, $\*n_{t} :=  \*N_{t} \mathbf{s}_k^* \sim \mathcal{CN}(\*0_{N\times 1},\sigma^2 \*I)$ is the equivalent noise vector at the antenna array at the output of the code-matched filter, i.e., such that $y_t$ has a raw SNR equal to $P/\sigma^2$ when no beamforming is applied.

In many practical scenarios only a partial information about $y_t$ is available to the BS. As a result, we consider the available signal to BS, $z_t$, to be of the form
\begin{equation} \label{eq:quanti}
    z_t = q(y_t),
\end{equation}
where $q(\cdot)$ represents a practically motivated  partial information processing such as a quantization function. With the received signal model in (\ref{eq:obsv}) and (\ref{eq:quanti}), we are now ready to describe the sequential beam search problem which adaptively designs the beamforming vectors $\*w_t$.

%For a cheap RF hardware, we can also consider a 1-bit observation
%\begin{equation}
%    y_t = \mathds{1} \[( |\sqrt{P} \*w_t^H \*h s_t + \*w_t \*n_t|^2 > u_t \]), \label{eq:obsv2}
%\end{equation}
%where the threshold is chosen together with $\*w_t$.

%\subsection{Single-Path Channel Model}

%% quest: shall we include this scope of the assumptions? 
%We assume that the large-scale channel $\sqrt{P}$ is known to the BS, and that the raw SNR at each antenna, $i.e.$ $ \frac{P}{\sigma^2}$, is relatively low ($-5$ to $10$ dB) due to the high path-loss nature of mmWave communication. Here we allow for any general model for the distribution of $\alpha$ (e.g. Rayleigh fading $\alpha \sim \mathcal{CN}(\mu_{\alpha}, \sigma^2_{\alpha})$). In our analysis, we assume that the realization of $\alpha$ is known to the BS. While for the numerical simulation in Section \ref{sec:num}, we also study the scenario where the BS does not know the realization of $\alpha$ but only its statistics (in effect using a mismatched fading coefficient).

% quest: new paragraph?
\section{Active Learning and Hierarchical Posterior Matching}
In this section we present our main result. In subsection~\ref{beam_alignment} we lay out the framework of active learning for sequential beam alignment. In subsection~\ref{Codebook} we describe the hierarchical beamforming codebook. In subsection~\ref{sect_hiePM} we describe our proposed algorithm: Hierarchical Posterior Matching for sequentially selecting the beamforming vector from the beamforming codebook. Lastly, in subsection~\ref{sec:posterior_update} we describe the posterior update for various measurement models.

\subsection{Sequential Beam Alignment via Active Learning}
\label{beam_alignment}
A sequential beam alignment problem in the initial access phase consists of a beamforming design strategy (possibly adaptive), a stopping time $\tau$, and a final beamforming vector design. Specifically, we consider a stationary beamforming strategy as a causal (possibly random) mapping function from past observations to the beamforming vector: $\*w_{t+1} = \gamma(z_{1:t},\*w_{1:t})$. Subsequently, the final beamforming vector selection $b(\cdot)$ is a (possibly random) mapping determining the final beamforming vector to be exploited for communication, $\hat{\*w} = b(z_{1:\tau},\*w_{1:\tau})$, as a function of the sequence of the observations gathered during the initial access phase $[1:\tau]$. To reduce the reconfiguration time of the beamforming vector from $\*w_{t}$ to $\*w_{t+1}$, we use a pre-designed beamforming codebook:
\begin{assumption}
The beamforming vector is chosen from a pre-designed beamforming codebook $\mathcal{W}^S$ with finite cardinality. 
\end{assumption}
 %This adaptive beamforming procedure is illustrated in Fig. \ref{fig:adaptive_beamforming}. 
Based on the nature of the protocol, we consider two criteria for selection of the length of the initial access phase:

%\subsection{Adaptive beamforming design}

%The observation in \ref{eq:obsv} is collected in a sequential and adaptive manner where the causal $y_{1:t-1}$ are used to design the next probing/searching beamforming $\*w_{t+1} = \gamma( y_{1:t},\*w_{1:t} )$ for collecting $y_{t+1}$. 

%$$ {\color{red} \*w_t} $$
%$$  y_t = \sqrt{P} {\color{red} \*w_t^H } \*a(\phi) s_t + {\color{red} \*w_t} \*n_t $$
%$$ y_{1:t} \rightarrow {\color{red} \*w_{t+1}} \in \mathcal{W}$$
%$$ \phi $$

%Let $\tau$ be the stopping time of the initial access phase where the BS outputs a beamforming vector $\hat{\*w}$ at time $\tau$ for future communication. 

\textbf{Fixed-length stopping time}: the user transmits a pre-determined number of frames during which the BS uses the beamforming vectors $\*w_1,\*w_2,...,\*w_T$. After the total pre-determined number of frames, $n$, the BS makes a prompt decision on the final beamforming vector $\hat{\*w}$

\textbf{Variable-length stopping time}: the user sends out the initial access signal continually until a certain target link quality can be achieved by the final beamforming vector $\hat{\*w}$ with high probability. Under a variable-length setup, the BS sends an ACK to the user which ends the initial access phase.

In Sec.~\ref{sec:hie_main}, we propose an adaptive beam alignment algorithm with both types of the stopping rules, while our analysis in subsection~\ref{analysis} focuses on the variable-length stopping time $\tau$. Our numerical studies consider the performance under both stopping rules.

 %We also consider two classes of beam search strategies: \textbf{Non-adaptive Beam Search} and \textbf{Adaptive Beam Search}.

% By the one-path channel modelling in the previous section, the estimation of $\phi$ is enough for designing the final beamforming $\hat{\*w}$. We first note that the point estimate $\hat{\phi}$ alone is not enough for a designing a good beamforming. If we simply use $\hat{\*w} = \*a(\hat{\phi})$, the SNR gain by this beamforming $|\hat{\*w}^H \*a(\phi)|^2$ may be quite small if $|\hat{\phi}-\phi|>\frac{\pi}{N}$ (see Fig. \ref{fig:} for illustration). The posterior distribution of $\phi$, $f(\phi|y_{1:\tau},\*w_{1:\tau})$, on the other hand, is sufficient for the design of $\hat{\*w}$. For instance, the beamforming vector that achieves the highest expected SNR gain is given by
% \begin{equation} \label{eq:bestbeam}
% \begin{aligned}
%      &\hat{\*w}_{\text{opt}}(y_{1:\tau},\*w_{1:\tau}) \\
%                                                      &= \argmax_{w}\int|\*w^H \*a(\phi')|^2 f(\phi'| y_{1:\tau},\*w_{1:\tau}) \ d\phi'. 
% \end{aligned}
% \end{equation}

%\subsection{AoA resolution and reliability}
Since the best beamforming vector $\hat{\*w}=\*a(\phi)$ can boost the SNR by a factor of $N$, the fading coefficient $\alpha$ can also be estimated and equalized easily if the SNR at the RF chain (after antenna combining) is high enough. Therefore, under Assumption~\ref{singlepath_alpha}, one of the major goals of the initial access phase is to learn the AoA $\phi$ so that BS can form a good beam toward that direction. Therefore, we can treat the sequential beam alignment problem by the methods of active learning \cite{Dasgupta2007,Naghshvar2015a} as shown in Fig. \ref{fig:adaptive_beamforming}, where the beamforming vector $\*w_t$ is equivalent to the query point and $y_t$ is equivalent to the response in the learning problem. The adaptivity of $\*w_t$ reflects that the query points are actively chosen as considered in active learning tasks.  

The quality of the established link, under a single-path channel model $\*h = \alpha \*a(\phi)$, is determined by the accuracy of the final point estimate, $\hat{\phi}(y_{1:\tau},\*w_{1:\tau})$, of $\phi$. In particular, a point estimate $\hat{\phi}$ together with a confidence interval $\delta$ provides robust beamforming with a certain outage probability. Hence, we measure the performance by the resolution and reliability of the final estimate $\hat{\*w}$:
\begin{definition}
Under Assumption~\ref{singlepath_alpha}, a sequential beam search strategy with an adaptive beamforming design $\gamma$, stopping time $\tau$, and final AoA estimate $\hat{\phi}$ is said to have resolution $\frac{1}{\delta}$ with error probability $\epsilon$ if
\begin{equation}
    \prob{|\hat{\phi}-\phi| > \delta} \leq \epsilon.
\end{equation}
\end{definition}
We note that, given a sufficiently large number of antennas, one can increase the resolution $1\slash\delta$ and decrease the error probability $\epsilon$ by increasing the time of sample collection $\tau$, or equivalently, by prolonging the initial access phase. In other words, the effectiveness of an initial access algorithm shall also be measured by the expected number of samples $\tau_{\epsilon,\delta}$ necessary to ensure a resolution $1\slash\delta$ and error probability $\epsilon$. From an information theoretic viewpoint, one can think of a family of sequential adaptive initial access schemes that achieves acquisition rate $R$ and reliability $E$:
\begin{definition}
Under Assumption~\ref{singlepath_alpha}, a family of sequential adaptive initial access schemes achieves acquisition rate-reliability (R,E) if and only if
\begin{equation}
    R:= \lim_{\delta \rightarrow 0} \frac{ \log (\frac{1}{\delta}) }{\expect{\tau_{\epsilon,\delta}}}, \quad E := \lim_{\epsilon \rightarrow 0} \frac{ \log (\frac{1}{\epsilon}) }{\expect{\tau_{\epsilon,\delta}}}.
\end{equation}
\end{definition}
\begin{remark}
The final beamforming vector (hence the quality of the established link) is determined by both the target resolution and the error $(\delta,\epsilon)$, and is written as ${\hat{\*w}(z_{1:\tau},\*w_{1:\tau},\epsilon,\delta)}$. Given a total communication time frame $T$, the expected spectral efficiency, under the final beamforming vector $\hat{\*w}$, is given as
\begin{equation} \label{eq:data_rate}
    \expect*{\frac{T-\tau}{T} \log \[( 1+  \frac{P \mid \alpha \hat{\*w}(z_{1:\tau},\*w_{1:\tau},\epsilon,\delta)^H \*a(\phi) \mid^2 }{\sigma^2}  \]) },
    %\label{rate}
\end{equation}
and is an important performance metric from a system point of view. This performance metric, however, requires a further system optimization over the length of the initial access phase, $\tau$, and the length of the communication phase, $T$, which is outside the scope of this paper. Therefore, in our analysis we focus on the parameters $\epsilon$ and $\delta$. For a comparison of different initial beam search algorithms, the system performance of (\ref{eq:data_rate}) is also evaluated in the numerical simulations for some nominal choice of $\tau$ and $T$.
\end{remark}

% We can usually get any resolution $\frac{1}{\delta}$ and error probability $\epsilon$ by waiting longer, with a variable length stopping time $\tau$ if the number of the antenna is larger enough. We will analyze the performance of a sequential beam search strategy with variable length stopping $\tau$ by upper bounding the expected variable length stopping time $\expect{\tau}$ as a function of $\epsilon$ and $\delta$.

%\begin{definition} \label{def:rate-reliability}
%Given a sequential beam search strategy $\Gamma$ with $(\gamma_n, \tau_n, \hat{\phi}_n, \delta_n, \epsilon_n)$ that has resolution $\frac{1}{\delta_n}$, error probability $\epsilon_n$ and a variable-length stopping time $\tau_n$ which satisfies $\expect{\tau_n} \leq n$. We say $\Gamma$ achieves acquisition rate $R$ and error error exponent $E$ if
%\begin{equation}
%    \lim_{n\rightarrow \infty} \frac{ \log (\frac{1}{\delta_n}) }{n} \geq R, \quad \lim_{n\rightarrow \infty} \frac{ \log (\frac{1}{\epsilon_n}) }{n} \geq E.
%\end{equation}
%\end{definition}

\subsection{Hierarchical Beamforming Codebook}
\label{Codebook}
We adopt the hierarchical beamforming codebook $\mathcal{W}^S$ proposed in \cite{Alkhateeb2014} with S levels of beam patterns. The beams divide the space dyadically in a hierarchical manner such that the disjoint union of the beams in each level is the whole region of interest. The codebook is the set $\mathcal{W}^S = \cup_{l=1}^S \mathcal{W}_l$, where $\mathcal{W}_l$ is all the beam patterns in level $l$ whose main beam has a width $\frac{|\theta_{\text{max}} - \theta_{\text{min}}|}{2^l}$. More specifically, for each level $l$, $\mathcal{W}_l$ contains $2^l$ beamforming vectors which divide the sector $[\theta_{\text{max}},\theta_{\text{min}}]$ into $2^l$ directions, i.e.
\begin{equation}
    [\theta_{\text{max}},\theta_{\text{min}}] = \bigcup_{k=1}^{2^l} D_{l}^k,
\end{equation}
each associated with a certain range of AoA $D_{l}^k$. The beamforming vector $\*w(D_{l}^k)$ is designed such that the beamforming gain $|\*w(D_{l}^k)^H \*a(\phi)|$ is almost constant for an AoA $\phi \in D_{l}^k$ and almost zero for $\phi \notin D_{l}^k$. %The total power is normalized to 1, i.e. $\*w^H \*w = 1$ for all $\*w \in \mathcal{W}^S$.

%given by 
%\begin{equation}
%   D_{s}^k = \{\theta_i:  i= (k-1) 2^{S-s} +1, ..., k 2^{S-s}  \}.
%\end{equation}
Note that $\mathcal{W}^S$ can be represented as a binary hierarchical tree, where each level-$l$ beam has two descendants in level $l+1$ such that each level-$l$ beam is the union of two disjoint beams, i.e., $D_{l}^k =  D_{l+1}^{2k}  \cup D_{l+1}^{2k-1}$. This binary tree hierarchy is illustrated in Fig.~\ref{fig:hie_beam} with the beam patterns of the first two levels of the codebook. Note that without loss of generality, the beamforming vectors in the codebook are assumed to have unit norm $\|\*w\|^2=1$. 
\begin{figure} 
    \centering
    \includegraphics[width = 0.48\textwidth]{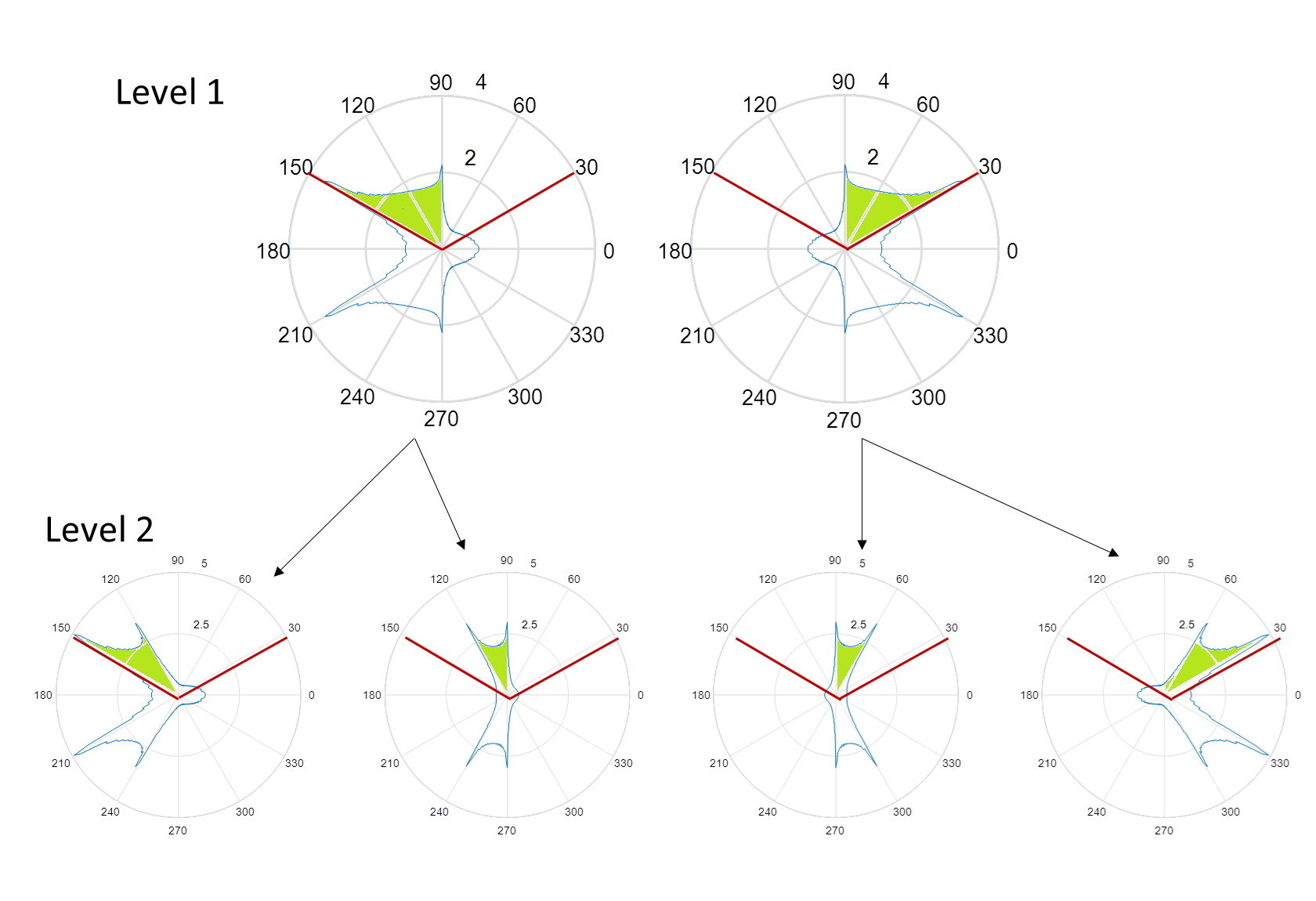}
    \caption{ The first 2 levels of hierarchical beamforming codebook with practical beam pattern formed by 64 antenna}
    \label{fig:hie_beam}
\end{figure}
%Also as illustrated in Fig. \ref{fig:hie_beam}, the codebook has a tree structure that each codeword $\*w(D_{s}^k)$ in level $s$ has two descendants in level $s+1$ that further partition $D_{s}^k$ into two sets of AoA $D_{s+1}^{k'}$ and $D_{s+1}^{k'+1}$. 
\subsection{Hierarchical Posterior Matching}
\label{sect_hiePM}
In this section we propose a search mechanism based on the connection between initial access beamforming, noisy search \cite{Chiu2016}, active learning \cite{Naghshvar2015a}, and channel coding with feedback \cite{Lalitha2017}, with the caveat that the beamforming vector is constrained to the practically feasible beamforming codebook of \cite{Alkhateeb2014} as in set $\mathcal{W}^S$. Instead of using all past observations $\*w_{t+1} = \gamma(z_{1:t},\*w_{1:t})$, $hiePM$ selects $\*w_{t+1}$ based on the posterior of the AoA $\phi$ at time $t$, which is a sufficient statistic. We discretize the problem by assuming that the resolution $\frac{1}{\delta}$ is an integer and that the AoA $\phi$ is from 
\begin{equation}
\phi \in \{ \theta_1,...,\theta_{1\slash \delta} \},\  \theta_i = \theta_{\text{min}} + (i-1)\times \delta \times (\theta_{\text{max}}-\theta_{\text{min}}).
\end{equation}
Such discretization approaches the original problem of initial access as $\delta \rightarrow 0$. To support resolution $1 \slash \delta$, the corresponding size of the hierarchical beamforming codebook 
\begin{equation}
    S = \log_2(1\slash \delta)
\end{equation}
is used. With this discretization, the posterior probability distribution can be written as a $\frac{1}{\delta}$-dimensional vector $\.\pi(t)$, where the $i^{th}$ component is of the form
\begin{equation}
    \pi_i(t) := \prob{\phi = \theta_i| z_{1:t}, \*w_{1:t}}, \ i=1,2,..., \frac{1}{\delta}.
\end{equation}
The posterior probability of the AoA $\phi$ being within a certain range, say $D_l^k$, can be computed as
\begin{equation}
    \pi_{D_l^k}(t) := \sum_{\theta_i \in D_l^k} \pi_i(t).
\end{equation}
%On the other hand, the resolution offered by the codeword $\*w(D_l^k)$ is given by the size of the region 
%\begin{equation}
%    |D_l^k| := \max_{\theta\in D_l^k} \theta - \min_{\theta\in D_l^k} \theta.
%\end{equation}

Now we are ready to present the proposed $hiePM$ algorithm. The proposed adaptive beamforming strategy, $hiePM$, chooses a beamforming vector at each time $t$ from the hierarchical beamforming codebook $\mathcal{W}^S$. The main idea of $hiePM$ is to select $\*w_{t+1}\in \mathcal{W}^S$ by examining the posterior probability $\pi_{D_l^k}(t)$ for all $l=1,2,...,S$ and $k=1,2,...,2^l$. Specifically, let
\begin{equation}
    l^*_t = \argmax_l \[\{ \max_k \pi_{D^k_l} \geq \frac{1}{2}  \]\},
\end{equation}
the proposed $hiePM$ algorithm selects a codeword at either level $l^*_t$ or $l^*_t+1$ according to Alg.~1 below. Given a snapshot of the posterior at time $t$, the selection rule is illustrated in Fig.~\ref{fig:hiePM}. The algorithm runs for a fixed length of time (fixed-length stopping) or until a certain error probability $\epsilon$ for resolution $1\slash \delta$ is achieved (variable-length stopping). The final choice of beamforming vector is determined by $\epsilon$ and $\delta$. The details of $hiePM$ are summarized in Algorithm 1 below.

%of the range of AoA associated with the corresponding codeword in $\mathcal{W}^S$ from the top level $l=1$ downward. As long as the AoA Posterior $D_{l}^k$ is larger than $\frac{1}{2}$, we further look into its two descendants at level $l+1$ and see if one of them is larger than $\frac{1}{2}$. This downward examination will continue until both of the descendants at certain level $l^*_t+1$ are both smaller than $\frac{1}{2}$. Then 
\begin{figure}
    \centering
    \includegraphics[width = 0.4\textwidth]{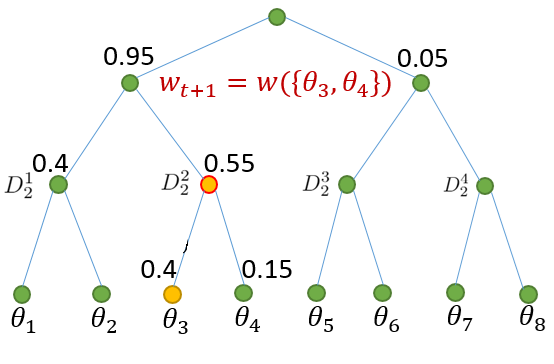}
    \caption{Illustration of the hierarchical posterior matching algorithm. In this example, we search down the tree hierarchy to levels 2 and 3, where level 3 has the first codeword that contains posterior lesser than half. Between level 2 and level 3, the codeword in level 2 of posterior $0.55$ is selected since it's closer to half ($0.55$ v.s. $0.4$). } %$l_{t+1} = 2$   $D_2^1$ $D_2^2$ $D_2^3$ $D_2^4$
    \label{fig:hiePM}
\end{figure}

\begin{algorithm}[h!tb] 
%Notation: $\pi_D(t) := \sum_{\theta_i\in D} \pi_i(t)$\\
 \textbf{Input}: target resolution $\frac{1}{\delta}$, target error probability $\epsilon$, codebook $\mathcal{W}^S$ $(S = \log_2(1\slash \delta))$, \textit{stopping-criterion} (with stopping time $n$ if fixed-length), \textit{algorithm-type} \\
 \textbf{Output}: Estimate of the AoA $\hat{\phi} $\\
 \textbf{Initialization}: $\pi_i(0) = \delta$ for all $i=1,2,...,1\slash \delta$, \\
 \For{$t=1,2,...$ }{
    
    \# Codeword selection from $\mathcal{W}^S$:\\
    $k=0$\;
    \For{$l=1,2,...,S$}{
        \eIf{$\pi_{D_l^k}(t)>1\slash 2$}{
            \# select the larger descendent\\ 
            $l_t^* = l$\; 
            $k \leftarrow \argmax_{k'\in \{2k,2k-1\}} \.\pi_{D_{l+1}^{k'}}(t)$\;
            %$l \leftarrow l + 1$;  
            }{  
            \vspace{-4mm}
            \begin{gather} 
            \begin{aligned}
                &(l_{t+1},k_{t+1}) = \\
                &\argmin_{(l',k') \in \{(l_t^*,\lceil{\frac{k}{2}}\rceil),(l_t^*+1,k)\}} \[|\pi_{D^{k'}_{l'}}(t) - \frac{1}{2}\]|
            \end{aligned}
            \end{gather}
            break\;
        }
    }
    %\While{$\pi_{D_{l_t^*}}(t) > 1\slash 2$ and $l_t^*<S$}{
    %$D_{l_t^*+1}(t) := \argmax_{D=D_{l_t^*}^{\text{left}}(t),D_{l_t^*}^{\text{right}}(t)}  \.\pi_D(t)$
    %$l_t^* \leftarrow l_t^* + 1$ \\
    %// $(l^*_t = \min \{ l:  1\leq l \leq S, \pi_{D_l}(t) \leq 1\slash 2 \})$  \\
    %}
    %$l_t = \argmin_{l=l^*_t,l^*_t+1} |\pi_{D_l(t)}(t) - 1\slash 2|$\\
    \setlength{\abovedisplayskip}{0pt} \setlength{\abovedisplayshortskip}{0pt}
    \# Codeword selection result \\[-2mm]
     \begin{gather}
         \*w_{t+1}  =  \*w(D_{l_{t+1}}^{k_{t+1}})
     \end{gather}\\
    \# Take next measurement \\[-2mm]
     \begin{gather} 
     \begin{aligned} \label{eq:measure}
         y_{t+1} &= \alpha \sqrt{P} \*w_{t+1}^H \*a(\phi) + \*w_{t+1}^H \*n_{t+1} \\
         z_{t+1} &= q(y_{t+1}) \\
     \end{aligned}
     \end{gather}\\
    \# Posterior update by Bayes' Rule   (Sec.~\ref{sec:posterior_update}) \\[-2mm]
     \begin{gather} 
     \label{eq:Bay}
         \.\pi(t+1) \leftarrow z_{t+1}, \.\pi(t) 
     \end{gather} \\
    %\# Stopping criteria \\
    \textit{case: stopping-criterion} = fixed length (FL)\\
    \If{ $t+1=n$}{
    %$\hat{\phi} = \argmax_{\theta_i}  \pi_i(n) $\\
    break (to final beamforming);}
    
    \textit{case: stopping-criterion} = variable length (VL)\\
    \If{ $ \max_i  \pi_i(t+1) > 1-\epsilon$}{
    %$\hat{\phi} = \argmax_{\theta_i}  \pi_i(t+1) $\\
    break (to final beamforming);}
    %\begin{equation*}
%\pi_i(t+1) = \frac{ \pi_i(t) f(y_{t+1}|\phi=\theta_i,\*w_{t+1})  }{   \sum_{j=1}^{1\slash\delta}\pi_j(t) f(y_{t+1}|\phi=\theta_j,\*w_{t+1}) }
%\end{equation*}\\
 }
     \# Final beamforming vector design \\
    $\tau = t+1$  (length of the initial access phase)\\
    \textit{case: algorithm-type} = fixed resolution (FR)\\[-2mm]
         \begin{gather}\label{eq:Fres}
         (\hat{l}, \hat{k}) =  (S, \argmax_k \pi_{D^k_{S}} (\tau) )
     \end{gather}\\
    \textit{case: algorithm-type} = variable resolution (VR)\\[-2mm]
         \begin{gather}\label{eq:Vres}
         \begin{aligned}
         \hat{l} &= \begin{cases}
         1,  \ \max_k \pi_{D^{\hat{k}}_1} (\tau) < 1-\epsilon \\
         \max \{ l:  \max_k \pi_{D^{\hat{k}}_l} (\tau) \geq 1-\epsilon \}, \  o.w.
         \end{cases} \\
         \hat{k} &= \argmax_k \pi_{D^k_{\hat{l}}} (\tau) \\
         \end{aligned}
         \end{gather}\\[-2mm]
    $\hat{\*w} = \*w(D^{\hat{k}}_{\hat{l}})$
\caption{Hierarchical Posterior Matching}
\end{algorithm}

%\begin{remark}
%By the $1\slash \delta$ discretization and the variable length stopping rule of reaching $\epsilon$ error described in Algorithm 1, the variable-length $hiePM$ has resolution $1\slash \delta$ with error probability $\epsilon$.
%\end{remark}

\begin{remark}
The $hiePM$ algorithm can be thought of as a noisy generalization of a bisection search where the posterior is used to create almost equally-probable search subsets subject to the codebook $\mathcal{W}^S$. Compared with the vanilla bisection method proposed in \cite{Alkhateeb2014}, $hiePM$ allows for significantly lower SNR search outcomes whose reliability are dealt with over time. This can also be viewed as water-filling in angular domain.
\end{remark}

\begin{remark}
See \cite{Lalitha2017,Shabara2017} for a detailed description of the connection between our beam search problem and a channel coding problem in data transmission. In this light, the vanilla noise-compensated bisection method of \cite{Alkhateeb2014} can be viewed as a repetition coding strategy which is known to have zero rate, while $hiePM$ can be viewed as a constrained (subject to hierarchical codebook $\mathcal{W}^S$) approximation to the capacity achieving posterior matching feedback coding scheme of \cite{Shayevitz2011}.
\end{remark}

\subsection{Posterior Update} \label{sec:posterior_update}
Let $\gamma_h: \.\pi(t) \rightarrow \mathcal{W}^S$ be the $hiePM$ sequential beamforming design given in Algorithm 1, i.e. let $\*w_{t+1} = \gamma_h(\.\pi(t))$. By the measurement model in (\ref{eq:measure}), the posterior update in Algorithm~1 in general can be written as
\begin{equation} \label{eq:Bayes}
\pi_i(t+1) = \frac{ \pi_i(t) f(z_{t+1}|\phi=\theta_i, \*w_{t+1} = \gamma_h(\.\pi(t)) )  }{  \sum_{j\neq i}\pi_j(t) f(z_{t+1}|\phi=\theta_j, \*w_{t+1} = \gamma_h(\.\pi(t)) )  },
\end{equation}
where $f(z_{t+1}|\phi=\theta_i, \*w_{t+1} = \gamma_h(\.\pi(t)) )$ is the conditional distribution of $z_{t+1}$ and depends on the function $q(\cdot)$ as well as the channel state information (e.g. the fading coefficient $\alpha$) known to the BS. Here, we give a few examples:

% $\tau_{\epsilon,\delta} = t$ //stopping time \\ 
%Given any $\delta>0$ and $\epsilon>0$, we can achieve resolution $\frac{1}{\delta}$ and error probability $\epsilon$ by the $\delta$-discretization together with the MAP estimate $\hat{\phi}=\max_i \pi_i(\tau)$ and a simple variable-length stopping time
%\begin{equation}
%    \tau_{\epsilon,\delta} := \min\{ t: 1-\max_i \pi_i(t) \leq \epsilon \}.
%\end{equation}
%Recall that the posterior with discretization is defined as
%\begin{equation}
%    \pi_i(t) := \prob{\phi = \theta_i| y_{1:t}, \*w_{1:t}}, \ i=1,2,..., \frac{1}{\delta},
%\end{equation}
%with a uniform prior $\pi_i(0)\equiv \delta$ and the Bayes' update
%where 
%\begin{equation}
%\begin{aligned}
%   f(y_t|&\phi=\theta_i,y_{1:t-1},\*w_{1:t-1}) \\
%   &= f(y_t|\phi=\theta_i,\*w_{t+1} = \gamma(y_{1:t},\*w_{1:t}))
%\end{aligned}
%\end{equation}
%depends on the channel model known to the BS.
% The variable stopping time $\tau_{\epsilon,\delta}$ in $hiePM$ can be written as
% \begin{equation}
%     \tau_{\epsilon,\delta} := \min\{ t: 1-\max_i \pi_i(t) \leq \epsilon \}.
% \end{equation}
% Together with the MAP estimate $\hat{\phi} = \argmax_{\theta_i}  \pi_i(t)$, the $hiePM$ has resolution $\frac{1}{\delta}$ and error probability $\epsilon$.

\begin{enumerate}
    \item \textit{Full measurement} $z_t=y_t$: 
    
    In the case of static channel (zero mobility), we may assume that the fading coefficient $\alpha$ is known to the BS. With a full measurement $z_t=y_t$, the conditional distribution of $z_t$ is a complex Gaussian, written as
    \begin{equation}
    \begin{aligned}
        f(z_{t+1} & |\phi=\theta_i, \*w_{t+1} = \gamma_h(\.\pi(t)) )  \\
        & = \mathcal{CN}(z_{t+1}; \alpha \sqrt{P} \*w^H_{t+1}\*a(\theta_i)  ,\sigma^2 ).
    \end{aligned}
    \end{equation}
    
     In the case where $\alpha$ is not known, the algorithm is assumed to use an estimate $\hat{\alpha}$:
     \begin{equation} 
    \begin{aligned}
        f(z_{t+1} & |\phi=\theta_i, \*w_{t+1} = \gamma_h(\.\pi(t)) )  \\
        & \approx \mathcal{CN}(z_{t+1}; \hat{\alpha} \sqrt{P} \*w^H_{t+1}\*a(\theta_i)  ,\sigma^2 )
    \end{aligned}
    \end{equation}
    for the posterior update. 
    
    %{\color{blue}
    %\item \textit{Amplitude only} $z_t=|y_t|$: 
    %When the fading gain $\alpha$ is either time-varying or the estimate of it is not accurate, the phase information can be misleading. In such case, we may consider using only the amplitude information as $z_t=|y_t|$. The corresponding density is written as 
    %\begin{equation} 
    %\begin{aligned}
    %    f(z_{t+1} & |\phi=\theta_i, \*w_{t+1} = \gamma_h(\.\pi(t)) )  \\
    %    & \approx \mathcal{CN}(z_{t+1}; \hat{\alpha} \sqrt{P} \*w^H_{t+1}\*a(\theta_i)  ,\sigma^2 )
    %\end{aligned}
    %\end{equation}
    %}
    
    \item \textit{1-bit measurement} $z_t = \mathds{1} \[(   |y_t| > v_t \])$:
    
    For practical high speed ADC implementation, we consider an extreme quantization function of a 1-bit \cite{Choi2016,Ding2018,Ho19} measurement model $z_t = \mathds{1} \[(   |y_t| > v_t \])$, where at each time instance $t$ the BS only has 1-bit of information indicating whether or not the received power passes the threshold $v_t$. Equivalently, we can write the measurement model as
    \begin{equation} \label{eq:zt_bern}
        z_t  = \mathds{1}(\phi\in D_{l_t}^{k_t} ) \oplus u_{t}(\phi), \ \ u_{t}(\phi) \sim \text{Bern}(p_t(\phi)),
    \end{equation}
where $u_{t}(\phi)$ is the equivalent Bernoulli noise with flipping probability $p_t(\phi)$, and $\oplus$ denotes the exclusive OR operator. The setting of the threshold $v_t$ and the corresponding flipping probability $p_{t}(\phi)$ is given in Lemma~\ref{lemma:opt_threshold}. In this case, the conditional distribution of $z_t$ can therefore be written as
\begin{equation}
    \begin{aligned}
        f(z_{t+1} & |\phi=\theta_i, \*w_{t+1} = \gamma_h(\.\pi(t)) )  \\
        & = \text{Bern}(z_{t+1} \oplus \mathds{1}(\theta_i\in D_{l_t}^{k_t} ); p_{t+1}(\theta_i) ). 
    \end{aligned}
\end{equation}

\end{enumerate}

\section{Analysis} \label{sec:hie_main}

\label{analysis}

Our analysis for $hiePM$ focuses on the variable-length stopping criteria with fixed resolution $\frac{1}{\delta}$ and a fixed target error probability $\epsilon$, where by Algorithm 1 the variable-length stopping time $\tau_{\epsilon,\delta}$ can be written as
\begin{equation}
    \tau_{\epsilon,\delta} = \min\{ t: 1-\max_i \pi_i(t) \leq \epsilon \}.
\end{equation}
We will also focus on the 1-bit measurement model described in Sec.~\ref{sec:posterior_update}. Furthermore, we make the assumption of an ideal hierarchical beamforming codebook for the analysis:
\begin{assumption} \label{asm:idealbeam}
The beam formed by the beamforming vector $\*w(D_{l}^k) \in \mathcal{W}^S$ has constant beamforming power gain for any signal of AoA $\phi\in D_{l}^k$ and rejects any signal outside of $D_{l}^k$, $i.e.$
\begin{equation} 
    |\*w(D_{l}^k)^H \*a(\phi)| = \begin{cases}  G_l, & \text{if } \phi \in D_{l}^k\\
    0, & \text{if } \phi \notin D_{l}^k
    \end{cases}.
\end{equation}
\end{assumption}
\begin{remark}
Assumption~\ref{asm:idealbeam} is mainly for better presentation. This assumption is approximately true when we have massive number of antennas $N\gg \frac{1}{\delta}$. The deterioration of performance due to the imperfect beamforming, such as that resulting from sidelobe leakage, is not the focus of our analysis. 
In our numerical simulations, however, we will remove this assumption by investigating the performance of the algorithms under the actual beamforming pattern with finite number of antennas.
\end{remark}

Under the 1-bit measurement $z_t = \mathds{1} \[(   |y_t| > v_t \])$ with Assumption~\ref{asm:idealbeam} and the optimal choice of the threshold $v_t$ in Lemma~\ref{lemma:opt_threshold}, the flipping probability $p_t(\phi)$ of the Bernoulli noise in (\ref{eq:zt_bern}) is independent of the AoA $\phi$ and only depends on the beamforming codeword level $l_t$ selected at time $t$. In particular, we have
\begin{equation}
    p_t(\phi) = p[l_t] := \int_0^{v_t} \text{Rice}(x; P G_l^2 ,\sigma^2)\ dx, 
\end{equation}
where $p[l] > p[l+1]$ and $ p[l] \rightarrow 0$ since $G_l < G_{l+1}$ and $G_l \rightarrow \infty$ as $l\rightarrow \infty$ (assuming unlimited number of antenna) by the design of the codebook. Furthermore, we assume that $\log_2(1\slash \delta)$ is an integer. Now we are ready to give an upper bound of the expected stopping time $\tau_{\epsilon,\delta}$ with resolution $\frac{1}{\delta}$ and outage probability $\epsilon$ of the proposed $hiePM$ sequential beam search algorithm:
\begin{theorem} \label{thm}
By using codebook $\mathcal{W}^S$ with $S = \log_2 (1\slash\delta)$ levels and assuming the perfect beamforming assumption (Assumption~\ref{asm:idealbeam}) and the 1-bit measurement model $z_t = \mathds{1} \[(   |y_t| > v_t \])$ with the optimal choice of $v_t$ in Lemma~\ref{lemma:opt_threshold}, the expected stopping time of $hiePM$, of resolution $\frac{1}{\delta}$ and error probability $\epsilon$, can be upper bounded by
\begin{equation}
\begin{aligned}
    \expect{\tau_{\epsilon,\delta}} \leq  \frac{\log (1\slash \delta)}{R_h} + \frac{\log (1\slash \epsilon)}{E_h}   + o( \log(\frac{1}{\delta \epsilon})),
\end{aligned}
\end{equation}
where $E_h = C_1(p[ \log_2 (1\slash\delta) ])$, $R_h = I(1\slash 3;p[l'])$ with $l'=\big\lfloor \frac{K_0 \lceil \log\log \frac{1}{\delta } \rceil}{\log2} -1 \big\rfloor$ and $K_0$ is a constant defined in Lemma~\ref{lemma:EJS_lq}.
\end{theorem}
\begin{proof}
See Appendix~\ref{sec:analysis-proof}
\end{proof}

\begin{corollary} \label{cor_eps}
Let, $\expect{\tau_{\epsilon,\delta}}=n$. For all values of $\delta$ such that $\delta \leq 2^{-nR_h}$, the error probability of $hiePM$ can be approximately upper bounded by
\begin{equation}\label{er}
    \prob{|\hat{\phi}-\phi| > \delta} \lessapprox \exp\[( - n E_h  \[(  1 - \frac{\log(1\slash \delta)}{nR_h} \])  \])
\end{equation}
when $\delta$ is small enough. 
\end{corollary}
\begin{corollary} \label{cor}
Under the same conditions and by Theorem~\ref{thm}, $hiePM$ achieves acquisition rate  
\begin{equation} \label{eq:hiePMrate}
\begin{aligned}
    \lim_{\delta \rightarrow 0}  \frac{\log (1 \slash \delta) }{\expect{\tau_{\epsilon,\delta}}} &\geq  \lim_{\delta \rightarrow 0} R_h \\
    &=\lim_{\delta \rightarrow 0} I(1\slash 3;p^*(\delta,\epsilon)) = 1
\end{aligned}
\end{equation}
for arbitrarily small error $\epsilon>0$, and error exponent
\begin{equation} \label{eq:hiePMexp}
  \lim_{\epsilon \rightarrow 0}  \frac{\log (1 \slash \epsilon) }{\expect{\tau_{\epsilon,\delta}}} \geq  \lim_{\epsilon \rightarrow 0} E_h = C_1(p[ \log_2 (1\slash\delta) ])
\end{equation}
for any $\delta>0$.
\end{corollary}
\begin{remark}
The integer assumption of $\log_2(1\slash \delta)$ is for notational simplicity. If the desired resolution $1\slash \delta$ is not of power of 2, one can simply take a higher resolution $ 1\slash \delta' = 2^{\lceil \log_2(1\slash \delta) \rceil} $. The corresponding upper bound in Theorem~\ref{thm} can be written accordingly and the conclusion in Corollary~\ref{cor} remains true.
\end{remark}

%By Theorem \ref{thm}, $hiePM$ effectively shrinks the range of AoA and spends most of the time on the signal quality at least $|\*w(D_{l_\alpha}^k)^H \*a(\phi)|^2$ when the target resolution $\frac{1}{\delta}$ is high. 
\begin{remark}
\label{remark:compare_rate}
The rate of one in equation (\ref{eq:hiePMrate}) implies that $hiePM$ performs asymptotically ($\delta \rightarrow 0$) in the same manner as a \underline{\textit{noiseless}} bisection search which is the optimal usage of the hierarchical beamforming codebook $\mathcal{W}^S$. The asymptotically noiseless behavior is due to the facts that $hiePM$ shrinks the AoA $D^k_l$ quickly, and that together with Assumption~\ref{asm:idealbeam} an unlimited number of antennas allow the beamforming gain $|\*w(D_{l}^k)^H \*a(\phi)|^2 = \frac{\pi}{|D^k_l|}\rightarrow \infty $ as $l\rightarrow \infty$. Compared with other beam alignment algorithms, non-adaptive random coding based strategies \cite{Abari2016} are not able to shrink the AoA region of the search beam. Therefore, the corresponding acquisition rate of \cite{Abari2016} rate is strictly lesser than 1. On the other hand, the adaptive noisy vanilla bisection algorithm in \cite{Alkhateeb2014} has rate zero even though the AoA region of the search beam shrinks over time. This is due to the fact that the noisy bisection of \cite{Alkhateeb2014}, in effect, employs repetition coding which has rate zero even with feedback (adaptivity). 
\end{remark}

\begin{remark}
To further compare our theoretical result of $hiePM$ with prior works, we plot Corollary \ref{cor} together with error probability upper bounds of \cite{Alkhateeb2014} and \cite{Abari2016} in Fig.~\ref{fig:upperbound} with $\expect{\tau}=28$, $1\slash\delta = 128$ and $|\theta_{\text{max}}-\theta_{\text{min}}| = 120^o$ and the ideal beamforming assumption (Assumption \ref{asm:idealbeam}). For the bisection algorithm of \cite{Alkhateeb2014}, we take the upper bound from the author's analysis for equal power allocation with fixed fading coefficient $\alpha=1$. While for the random hashing of \cite{Abari2016}, we take an optimization of the number of directions over Gallager's random coding bound of BSC as
\begin{equation}
    P_e \leq \min_{q} \exp\[( - 28 \times  E_{RC}(q) \]),
    \label{eq:rand_bound}
\end{equation}
where $ E_{RC}(q) = \max_{0\leq \rho\leq 1} \[( E_0(\rho,q) - \rho \times \frac{\log_2(128)}{28} \])$ and
\begin{equation}
\begin{aligned}
   E_0(\rho,q) = -& \log  \Big (\[( q (p_q)^{\frac{1}{1+\rho}} + (1-q) (1-p_q)^{\frac{1}{1+\rho}}   \])^{1+\rho}\\ 
   &+ \[( q (1-p_q)^{\frac{1}{1+\rho}} + (1-q) (p_q)^{\frac{1}{1+\rho}} )  \])^{1+\rho}\Big)
\end{aligned}
\end{equation}
with
\begin{equation}
    p(q) := \int_0^{v_t} \text{Rice}(x; P \frac{3}{2q} ,\sigma^2)\ dx, 
\end{equation}
where again $v_t$ is optimally chosen according to Lemma~\ref{lemma:opt_threshold}. The illustration of Corollary~\ref{cor} in Fig.~\ref{fig:upperbound} predicts the superior performance of $hiePM$ over the prior works \cite{Alkhateeb2014} and \cite{Abari2016}. We note that for these upper bounds, $hiePM$ and random hashing of \cite{Abari2016} assume a 1-bit quantizer, whereas the bisection method of \cite{Alkhateeb2014} is favorably given the unquantized amplitude information. We will further show in numerical simulation (Fig.~\ref{fig:sim-errorall}) that with practical beam patterns and unquantized measurements, the actual performance of $hiePM$ is not only indeed better than the prior works, but in fact achieves a significantly smaller error probability than our theoretical upper bound. Furthermore, we will see that the non-adaptive random hashing based method in \cite{Abari2016} in fact outperforms the adaptive bisection in \cite{Alkhateeb2014} due to the lack of good coding in \cite{Alkhateeb2014}.
\end{remark}

\begin{figure}
    \centering
    \includegraphics[width = 0.48\textwidth]{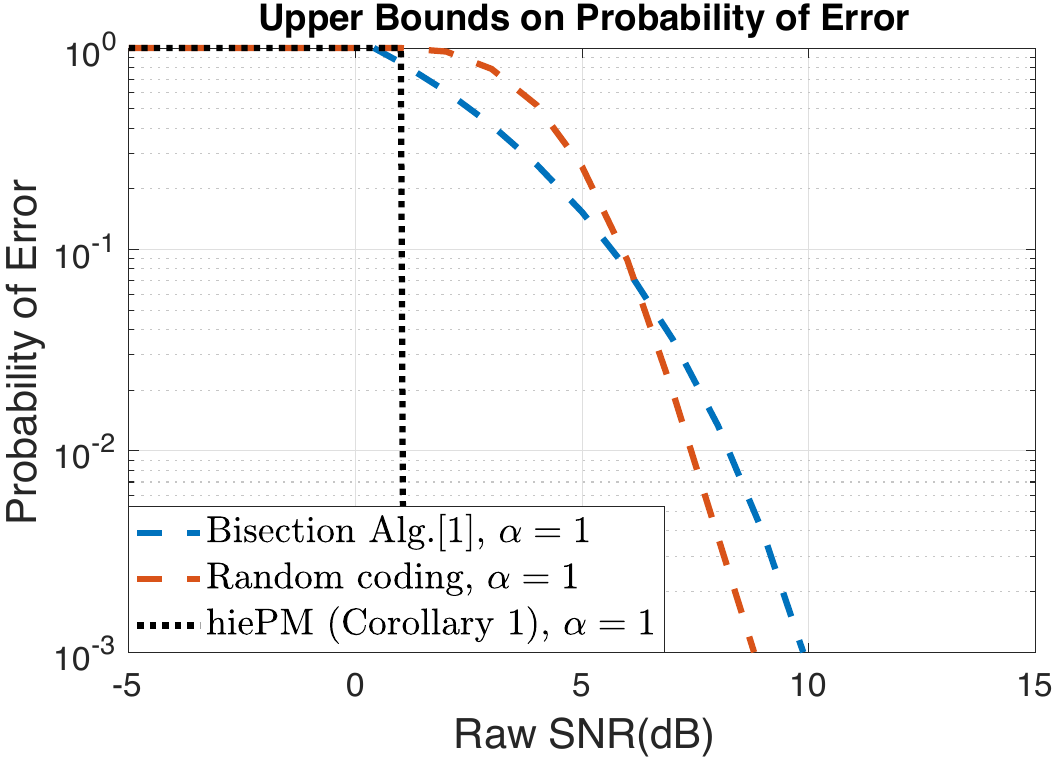}
    \caption{Comparison of the theoretical upper bounds on error probability between $hiePM$, random coding, and the bisection algorithm of \cite{Alkhateeb2014} as a function of raw SNR $P/\sigma^2$. The upper bound on $hiePM$ is given by Corollary~\ref{cor_eps}. While the upper bound on random coding is given by Gallager's bound as in~(\ref{eq:rand_bound}), and the upper bound on the bisection algorithm is provided by \cite{Alkhateeb2014}. } 
    \label{fig:upperbound}
\end{figure}

\section{Numerical Results}
\label{sec:num}

In this section, we compare the performance of our proposed \textit{hiePM} algorithms against the bisection algorithm of \cite{Alkhateeb2014} and an optimized random-code-based strategy, which is taken as an upper bound on the performance of the random hash-based solution of \cite{Abari2016}. Before we proceed with this performance analysis, however, we first provide a summary of the simulation setup and parameters.

\subsection{Simulation Setup and Parameters}
We use the hybrid analog/digital system architecture described in Sec.~\ref{System}, where the BS has $N = 64$ antenna elements in a uniform linear array with antenna spacing $\frac{\lambda}{2}$, and the user has a single (virtual) antenna. Furthermore, due to the use of orthogonal spreading sequences, we focus on the single user case $K=1$. The channel consists of a single path with fast fading coefficient $\alpha$. 
{The rule-of-thumb \cite{Rappaport2002} estimate of channel coherence time given by
\begin{equation}
    T_c  \approx \frac{0.432}{f_m} = \frac{0.432 c}{f_c v},
\end{equation}
where $c$ is the speed light, $f_c$ is the carrier frequency, and $v$ is the user speed. So even for mmWave communication with $73$GHz, at walking speed ($< 3$ km/hour) the coherence time is
\begin{equation}
T_c = \frac{0.432 \times 3\times 10^9 \text{ (m/s) } }{73\times 10^9  \text{ (Hz) }  \times 3 \text{ (km/hour) } } \approx 8.127 \text{ milliseconds}.
\end{equation}
Note that, additionally, narrow beamforming and the existence of a dominant sub-path (e.g. Line-of-Sight) can both increase the coherence time significantly \cite{Va2017}. Therefore, in subsection~\ref{results:static_alpha} we assume that the fading coefficient $\alpha$ is static during the entire initial access duration of $2$ milliseconds (ms).
} We consider both the case when the fading coefficient $\alpha$ is known exactly $\hat{\alpha}=\alpha$, and the case when it is estimated with the estimation inaccuracy modelled as $\hat{\alpha} \sim \mathcal{CN}({\alpha},\sigma_{\alpha}^2)$. 
{In subsection~\ref{results_timevarying_alpha} we further study the robustness of $hiePM$ with a static estimate of the time varying fading coefficient $\alpha_t$ of a Rician AR-1 model with a coherence time corresponding to higher user mobility.} Finally, we consider learning the AoA with an angular resolution of $1\slash \delta=128$, and an (expected) stopping time of $E[\tau]=28$, i.e., with $E[\tau]$ selections of beamforming vectors, hence, samples.

To provide a sense for the above normalized parameters, let us consider some candidate PHY layer solutions. In particular, when using the 5G new radio Physical Random Access Channel (PRACH) format B4 \cite{Lin20185GNR}, the $E[\tau]=28$ samples translate to less than 2 ms acquisition time for sub-1-degree angular resolution within a $[0\degree, 120\degree]$ sector. {We present our results as a function of raw SNR $\frac{P}{\sigma^2}$ to get a sense for reasonable values of SNR. In Fig.~\ref{fig:SNR-dist} we compute and illustrate the expected distance at which a target raw SNR is obtained.} We consider a case under the 3GPP TR 38.901 UMi LOS pathloss channel model (summarized in \cite{Rappaport20185GNRmodels}), with 23 dBm maximum user power, -174 dBm/Hz thermal noise density, 5 dB receiver noise figure at BS, with a bandwidth of 100 MHz. {As seen in Fig.~\ref{fig:SNR-dist}, one can argue that } given our selection of this PHY layer and parameters, the practical raw SNR regime of interest is within $-15$ dB to $10$ dB. 

%%%%%%%%%% SNR vs Dist %%%%%%%%%%

\begin{figure}[ht]
    \centering
    \includegraphics[width = 0.45\textwidth]{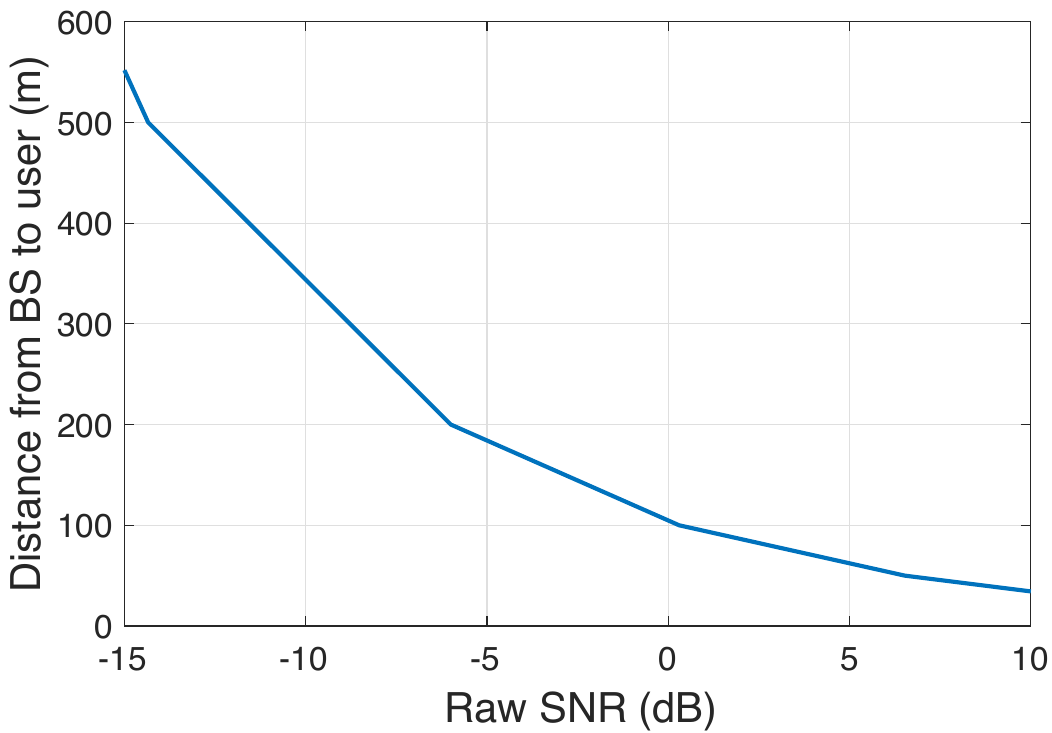}
    \caption{Relationship between raw SNR $P \slash \sigma^2$ and distance from BS to user, under the 3GPP TR 38.901 UMi LOS pathloss channel model (summarized in \cite{Rappaport20185GNRmodels}), with 73GHz carrier frequency, 23 dBm maximum user power, -174 dBm/Hz thermal noise density, 5 dB receiver noise figure at BS, and a bandwidth of 100 MHz.
    }
    \label{fig:SNR-dist}
\end{figure} 
%%%%%%%%%%%%%%%%%%%%%%%%%%%%%%%%%

\subsection{Algorithm Details and Parameters}
%Bisection algorithm explaination
%We take the bisection algorithm of \cite{Alkhateeb2014} to represent adaptive algorithms based on using beam refinement to sequentially scan the region of interest. %The bisection algorithm of \cite{Alkhateeb2014} implements the hierarchical codebook that we've summarized in Sect.~\ref{Codebook} and searches each partition two times in each stage. 
Like the the bisection algorithm of \cite{Alkhateeb2014}, our proposed algorithm \textit{hiePM} is based on sequential beam refinement, but implements additional coding techniques. Thus, we focus our comparison to the bisection refinement of \cite{Alkhateeb2014} to highlight that the use of this coding strategy differentiates \textit{hiePM} from existing beam refinement strategies. For the bisection algorithm of \cite{Alkhateeb2014}, the number of beamforming vectors in each level is $2$, and the power is allocated according to the equal power distribution strategy.

For both $hiePM$ and the bisection algorithm of \cite{Alkhateeb2014}, the finite set of beamforming vectors $\mathcal{W}^S$ are designed with a hierarchical structure,
where individual beamforming vectors $\mathbf{w}(\mathcal{D}_l^k) \in \mathcal{W}_l$ are designed with the objective of near constant gain for signal directions with AoA $\phi \in \mathcal{D}_l^k$ and zero otherwise (Assumption 4). In other words, each beamforming vector solves

\begin{equation}
    \mathbf{A}_{BS}^{H}\mathbf{w}(\mathcal{D}_l^k) = C_s \mathbf{G}_{\mathcal{D}_l^k},
    \label{weght_vecs_prob}
\end{equation}
where $\mathbf{A}_{BS}$ is the $N \times (1\slash \delta)$ matrix of array manifolds
\begin{equation}
    \mathbf{A}_{BS} = [\mathbf{a}(\phi_1) , \mathbf{a}(\phi_2), \ldots \mathbf{a}(\phi_{1 \slash \delta})]
    \label{A_BS}
\end{equation}
$C_s$ is a normalization constant, and $\mathbf{G}_{\mathcal{D}_l^k}$ is an $1\slash \delta \times 1$ vector indicating probed directions
\begin{equation} 
    \mathbf{G}_{\mathcal{D}_l^k}= \begin{cases}  1, & \text{if } \phi \in D_{l}^k\\
    0, & \text{if } \phi \notin D_{l}^k
    \end{cases}.
\end{equation}
An approximate solution to (\ref{weght_vecs_prob}), obtained using the pseudo inverse, is 
\begin{equation}
    \mathbf{w}(\mathcal{D}_l^k) = C_s (\mathbf{A}_{BS}\mathbf{A}_{BS}^{H})^{-1} \mathbf{A}_{BS}\mathbf{G}_{\mathcal{D}_l^k}.
    \label{weght_vecs_soln}
\end{equation}
The resulting beamforming weight vectors, applied with phase and gain control at each element, produce beam patterns with improved sidelobe suppression, and near constant gain in the intended probing directions. We can use these vectors in our simulations to ensure that our analytic Assumption~\ref{asm:idealbeam} is a matter of analytic simplicity but is not consequential in a practical setting. % %We note that the bisection algorithm of \cite{Alkhateeb2014} allows for $K>2$, and while \textit{hiePM} could be adapted for any $K$, the resulting codebook will differ in each case. Furthermore, in \cite{Alkhateeb2014} (Fig. 7) they show best performance in spectral efficiency for $K=2$ and thus we adopt this value for the entirety of our simulations. 
%Random Search Description

To represent non-adaptive algorithms that are a variation of random coding, such as the random hashing algorithm of \cite{Abari2016}, we compare to the random search algorithm that randomly scans the region of interest. The random search algorithm uses a codebook $\mathcal{W}^{\frac{q}{n}}$ of size $|\mathcal{W}^{\frac{q}{n}}|$ = $\binom{n}{q}$ which consists of all possible beam patterns with total width {$\frac{q}{n}|\theta_{\text{max}}-\theta_{\text{min}}|$}, where the region of interest {$|\theta_{\text{max}}-\theta_{\text{min}}|$} has been divided into $n$ non-overlapping directions, and $q$ directions are probed in each beam pattern. At any time instant $t$, the random search algorithm randomly selects a beamforming vector $\textbf{w}_{t+1}$ from the pre-designed codebook $\mathcal{W}^{\frac{q}{n}}$. A fixed number of measurements $\tau$ are taken according to~(\ref{eq:measure}) and the final beamforming vector is selected according to~(\ref{eq:Bay}) and~(\ref{eq:Fres}). The discretization parameter is set to $n= 1\slash\delta = 128$, we set $\tau = 28$,  and we plot various values of $q$. The performance of $hiePM$ over the optimized choice of $q$ is important as it provides a first order insight into ``adaptivity gain."

\subsection{Simulation Results}
\label{results:static_alpha}
In this section, we provide a comparative analysis of our proposed \textit{hiePM} algorithm against prior work \cite{Alkhateeb2014} and \cite{Abari2016}. In particular, Fig.~\ref{fig:sim-errorall} plots the error probability as a function of raw SNR. In summary, Fig.~\ref{fig:sim-errorall} shows that both fixed-length and variable-length stopping variations of \textit{hiePM} outperform the bisection algorithm of \cite{Alkhateeb2014} as well as random coding, or random-hash based solutions of \cite{Abari2016}. We also note that random beamforming codebooks outperform the bisection algorithm of \cite{Alkhateeb2014}, as expected by our analysis in Remark \ref{remark:compare_rate}. By optimizing the coding rate $q$, and comparing against \textit{hiePM}, one can also fully characterize the adaptivity gain. Finally, we note that our analytic upper bound {(in Fig.~\ref{fig:upperbound})} is rather loose and $hiePM$ performs significantly better {than our analysis predicted}.

\subsubsection{Probability of Error versus Raw SNR}
%\textit{hiePM}(FL, Fres) \textit{hiePM}(VL, Fres)
%\textit{hiePM}(FL, Vres)

%(\textit{hiePM}, FL, Fres)

%\textit{hiePM}(FL, FR)

%\textit{hiePM(fixed $\tau$, fixed $\delta$)}

%\textit{hiePM(variable $\tau$, fixed $\delta$)}

%\textit{hiePM(fixed $\tau$, variable $\delta$)}

%\textit{hiePM}($\tau$, $\delta$)  

For the system and channel described above, we conduct the simulation scenario where the average error probability as a function of raw SNR is analyzed. We take the error probability of the AoA estimation to be the probability of selecting an erroneous final beamforming vector $\text{Prob}\{\hat{\textbf{w}}(z_{1:\tau}, \textbf{w}_{1:\tau}, \epsilon, \delta) \neq \textbf{w}(\{\phi\})\}$.%For practicality, we can think of the raw SNR ($\frac{P}{\sigma^2}$) in terms of a physical basis, such as the distance between the BS and the user in meters. We adopt the 3GPP TR 38.901 UMi LOS pathloss channel model (summarized in \cite{Rappaport20185GNRmodels}), with 23 dBm maximum user power, -174 dBm/Hz thermal noise density, 5 dB receiver noise figure at BS, with a bandwidth of 100 MHz
%to illustrate this relationship in Fig.~\ref{fig:SNR-dist}. 

%In the resulting Fig. \ref{fig:sim-errorall}, the average probability of error is simulated for the channel conditions discussed above and plotted as a function of the raw SNR. We plot various case selections of Algorithm 1.

%%%%%%%%%% Avg Probability of Error %%%%%%%%%%

\begin{figure}
    %\centering
\includegraphics[width = 0.48\textwidth]{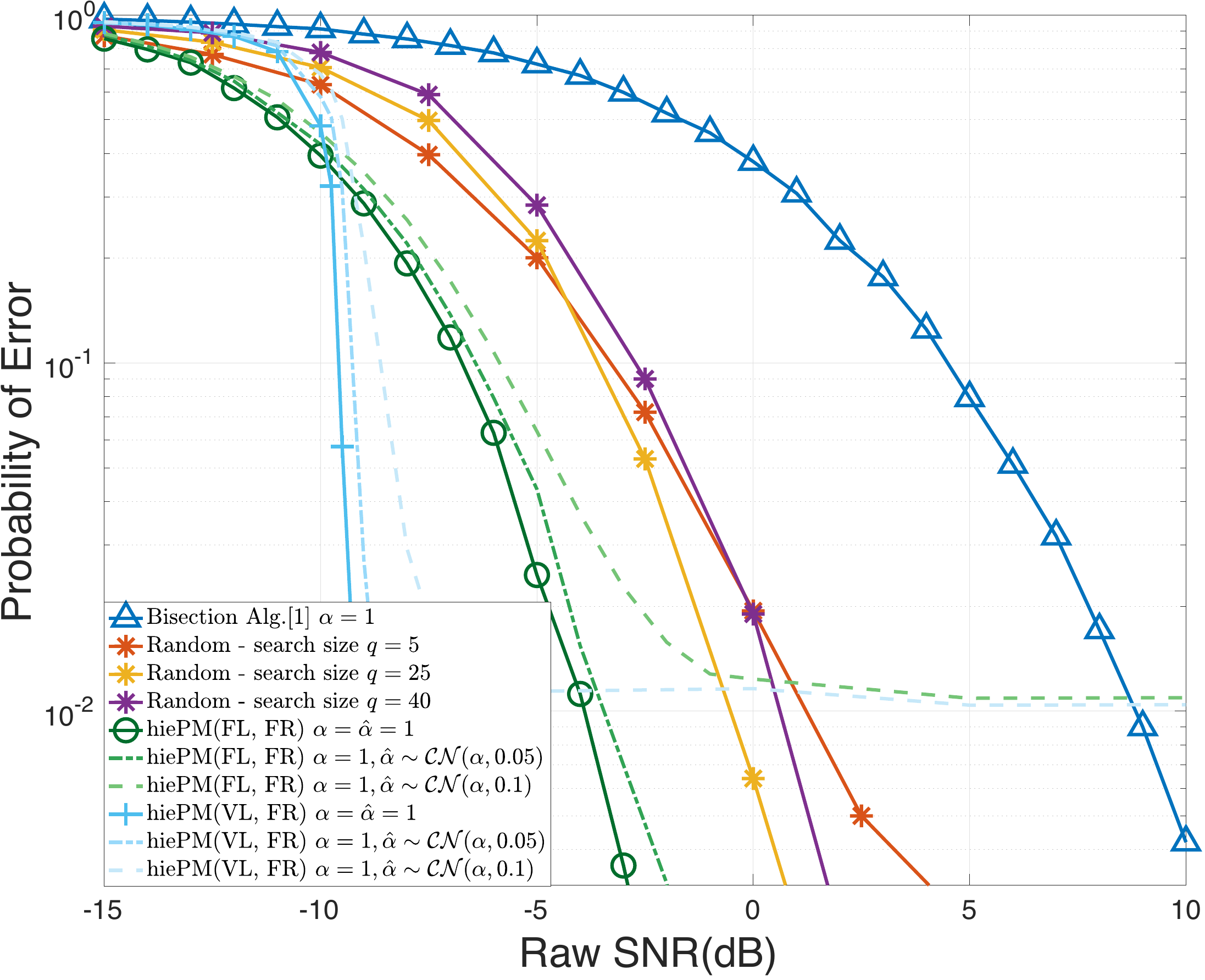}
    \caption{Comparison of the error probability between \textit{hiePM}, the random search algorithm, and the vanilla bisection algorithm of \cite{Alkhateeb2014} as a function of raw SNR  $P/\sigma^2$. Initial access length $\tau=28$, achieved under 2 ms using the 5G NR PRACH format B4 \cite{Lin20185GNR} ($E[\tau]=28$ for variable-length stopping type), is used for acquiring the AoA with resolution $1\slash \delta=128$. %The required SNR thresholds of $l$th level beam ($l=1$ to $l=5$ from right to left) are described in equation (\ref{eq:SNRthre})
    }
    \label{fig:sim-errorall}
\end{figure} 

For clarity, from now on we use the naming convention \textit{hiePM(stopping-criterion, resolution-criterion)} to specify the case selections of stopping criteria and resolution-criteria in the proposed \textit{hiePM} algorithm (detailed in algorithm 1). To ensure a reasonable comparison, we first discuss \textit{hiePM}(FL, FR) which is most comparable to the bisection algorithm of \cite{Alkhateeb2014} and the random search algorithm described above. Fig. \ref{fig:sim-errorall} shows the superior performance of \textit{hiePM}(FL, FR) with fixed and known fading coefficient $\alpha =1$ over both the bisection algorithm of \cite{Alkhateeb2014} and the random search algorithm. 
We also notice that under reasonable tuning of parameter $q$, even the non-adaptive random search algorithm achieves better performance than the adaptive bisection algorithm of \cite{Alkhateeb2014}. As we expected from Remark~\ref{remark:compare_rate}, the best performance is achieved by \textit{hiePM} due to its sequential coding strategy, while the performance of the bisection algorithm of \cite{Alkhateeb2014} suffers as it resembles a repetition code. %The undesirable performance of the the bisection algorithm of \cite{Alkhateeb2014} is even predicted by its upper bound, shown in Fig. \ref{fig:sim-errorall}.

Improvements in the probability of error are further demonstrated by \textit{hiePM}(VL, FR) with targeted error probability $\epsilon$ selected such that $\mathbb{E}[\tau] = 28$, { i.e., the parameter $\epsilon$ is strategically chosen (usually $\epsilon$ is large), such that \textit{hiePM}(VL, FR) has an expected duration of the initial access phase equal to the duration of the fixed-length variations, this ensures a fair comparison. Of course, $\epsilon$ may be selected to be close to zero, thereby ensuring the best error probability and aquisition rate at the cost of a longer initial access phase duration}. The benefit of allowing a variable stopping time is evident in that it causes a sharp drop in the error probability at approximately -10 dB raw SNR. The error probability upper bound (Corollary~\ref{cor_eps}) on \textit{hiePM}(VL, FR) is also plotted. We see in Fig.~\ref{fig:sim-errorall} that this upper bound predicts the sharp slope of \textit{hiePM}(VL, FR), theoretically guaranteeing a significant performance improvement in error probability for \textit{hiePM}(VL, FR) over the bisection algorithm of \cite{Alkhateeb2014} and the random search algorithm for large SNR. A further exploration of this sharp transition in the low ($<0$ dB) raw SNR regime is beyond the scope of this paper.

%Empirically we see that below this sharp cut-off point, the acquisition rate is below the rate the channel can sustain, thus there is poor performance in \textit{hiePM}(VL, FR) and there is a small benefit in using \textit{hiePM}(FL, FR) in this low SNR instead. This is likely due to the early sub-optimal stopping time when the parameter $\epsilon$ is too large, as well as the strong-converse \cite{Wolfowitz1968}. A further exploration of this sharp transition is beyond the scope of this paper.

\subsubsection{Investigating effects of imperfect channel knowledge}
The bisection algorithm of \cite{Alkhateeb2014} learns the AoA without any knowledge of the channel. It combines the procedures of AoA estimation and channel estimation.   
On the other hand our proposed algorithm \textit{hiePM}, requires knowledge of the fading coefficient $\alpha$ in the posterior update of Algorithm 1~(\ref{eq:Bay}).  
While a channel estimation procedure can be used to learn $\alpha$ preceding \textit{hiePM}, perhaps in a short preliminary phase, we explore the performance achieved using an estimate for the fading coefficient $\hat{\alpha}$ instead. We find that the improved performance by \textit{hiePM} over the bisection algorithm of \cite{Alkhateeb2014} and the random search algorithm holds even without full knowledge of the fading coefficient $\alpha$. %While we do not study the incorporation of estimating $\alpha$ and $\phi$ together in this paper,
To see this we consider the case of a mismatched update rule (\ref{eq:Bay}) with an estimate for the fading coefficient $ \hat{\alpha} = \mathcal{CN}(\alpha, \sigma_{\alpha}^2)$. We see that even under a reasonably mismatched estimate of the fading coefficient ($\sigma_{\alpha}^2=0.05$), all \textit{hiePM} based algorithms still achieve a lower probability of error than the bisection algorithm of \cite{Alkhateeb2014}. In other words, the degradation due to estimation error is far less significant, saturating in error probability only at very large SNR ($> 5$ dB). {As we can see in Fig.~\ref{fig:sim-errorall} using a mismatched estimate of the fading coefficient $\alpha$ causes the performance of probability of error to saturate at large SNR ($> 0$ dB). This reflects the times when the estimate of the fading coefficient $\alpha$ is very inaccurate, which will occur with a constant probability regardless of the SNR value, due to our modeling of $\hat{\alpha}$. In practice, the accuracy of the estimate $\hat{\alpha}$ will improve as SNR increases. However, this is beyond the scope of this paper and we refrain from investigating such effects.}

\label{PoEsims}

\subsubsection{{Spectral efficiency} versus Raw SNR}
Practically speaking, a more efficient AoA learning algorithm is advantageous in that it both reduces communication overhead and increases the accuracy of the final beamforming vector. Next, we empirically analyze the overall performance of a communication link established by the proposed algorithm \textit{hiePM} in terms of the data {spectral efficiency}. The {spectral efficiency} is evaluated according to equation (\ref{eq:data_rate}), using the final beamforming vector $\hat{\textbf{w}}(z_{1:\tau}, \textbf{w}_{1:\tau}, \epsilon, \delta)$ resulting from each algorithm. Due to its dependence on the final beamforming vector $\hat{\textbf{w}}$, the spectral efficiency encompasses both the design parameters $\epsilon$ and $\delta$, which have been the focus of our analysis, while still providing an intuitive practical measure. We set the total communication time frame to $T = 100 \times \mathbb{E}[\tau]$ (further optimization of this parameter beyond the scope of this paper).

%%%%%%%%%% Rate all %%%%%%%%%%
\begin{figure}
    \centering
    \includegraphics[width = 0.48\textwidth]{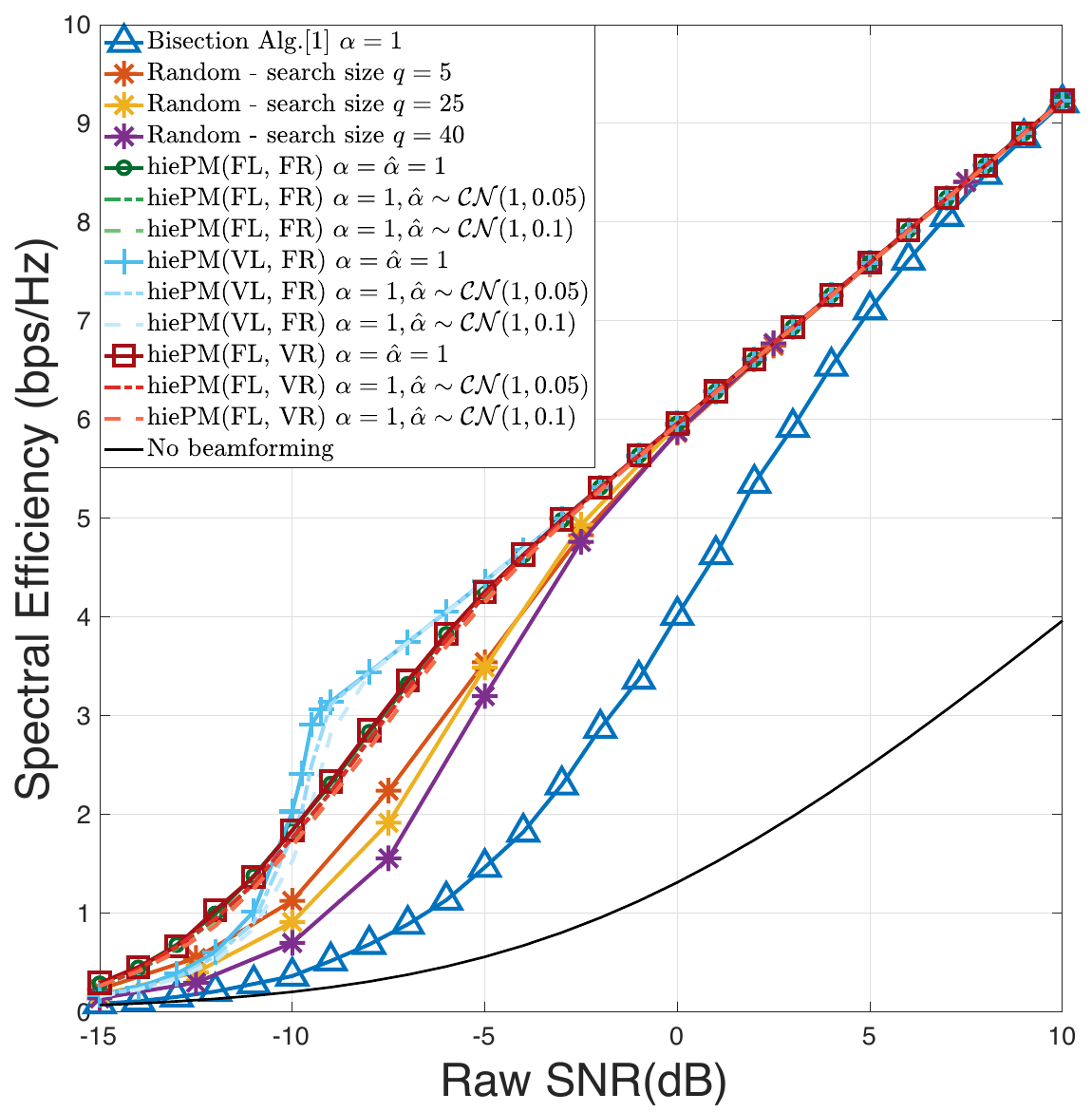}
    \caption{Comparison of the {spectral efficiency} obtained by \textit{hiePM}, the random search algorithm, and the vanilla bisection algorithm of \cite{Alkhateeb2014} as a function of raw SNR  $P/\sigma^2$. Initial access time $\tau=28$, achieved under 2 ms using the 5G NR PRACH format B4 \cite{Lin20185GNR} ($E[\tau]=28$ for variable-length stopping). The {spectral efficiency} is given by (\ref{eq:data_rate}) with the final beamforming vector $\hat{\textbf{w}}$ designed by the respective algorithm.}
    \label{fig:sim_rate}
\end{figure}

Fig. \ref{fig:sim_rate} shows the gain in {spectral efficiency} obtained by various implementations of the proposed algorithm \textit{hiePM} over both the bisection algorithm of \cite{Alkhateeb2014} and the random search algorithm for the system and channel described above as a function of raw SNR. The {spectral efficiency} when no beamforming is used is provided for reference. Fig. \ref{fig:sim_rate} shows that all variants of \textit{hiePM} outperform the bisection algorithm of \cite{Alkhateeb2014} significantly in the raw SNR regime of (-5dB to 5dB). On the other hand, the performance of the bisection algorithm of \cite{Alkhateeb2014} approaches the performance of \textit{hiePM} as raw SNR grows beyond 7dB or so.  Fig. \ref{fig:sim_rate} also shows the benefits of opportunistically selecting the resolution of the final beam as is done under \textit{hiePM}(FL, VR) according to (\ref{eq:Vres}). This is particularly important in very low raw SNR models (-15dB to -7dB) where \textit{hiePM}(FL, VR) adapts the final beamforming vector to the final posterior distribution at time $\tau$, hence setting the angular resolution of the communication beam in an opportunistic manner. Even more importantly, this significant performance improvement is robust to channel estimation error and mismatched estimate of the fading coefficient $\hat{\alpha}$. To understand this phenomenon we refer to Fig. \ref{fig:sim-errorall}, where the error probability of finding the correct beam with resolution $1/\delta$, when SNR is less than -5dB, is non-negligible under \textit{hiePM}(FL, FR), and \textit{hiePM}(VL, FR). 

\subsection{Time varying channel}
\label{results_timevarying_alpha}
\begin{figure}
    \centering
\includegraphics[width =0.48\textwidth]{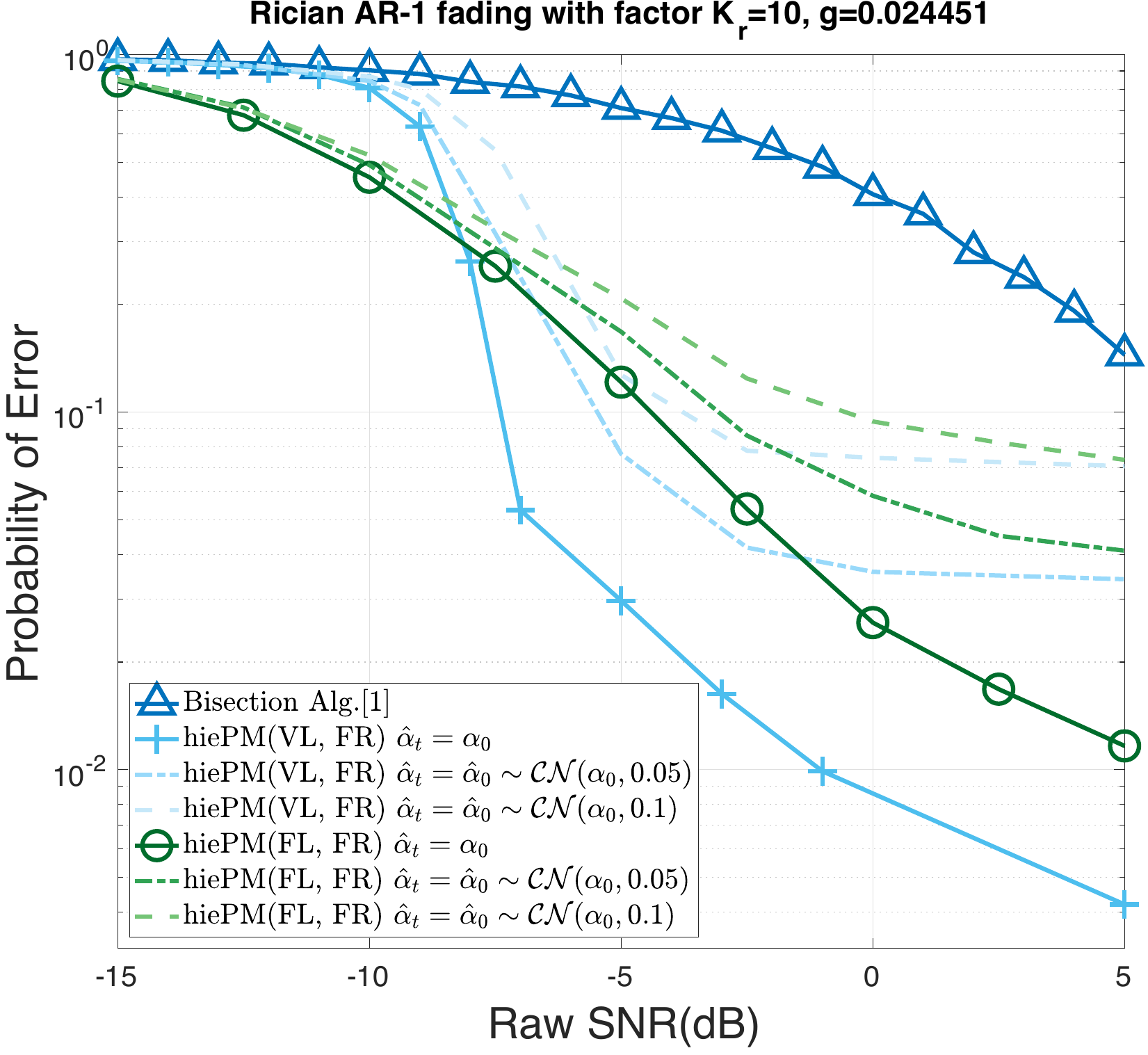}
    \caption{Comparison of the error probability between \textit{hiePM}, the random search algorithm, and the vanilla bisection algorithm of \cite{Alkhateeb2014} as a function of raw SNR  $P/\sigma^2$ under Rician AR-1 fading with factor $K_r =10$, and $g=0.024451$ (i.e. $T_c=2$). Initial access length $\tau=28$, achieved under 2 ms using the 5G NR PRACH format B4 \cite{Lin20185GNR} ($E[\tau]=28$ for variable-length stopping type), is used for acquiring the AoA with resolution $1\slash \delta=128$. %The required SNR thresholds of $l$th level beam ($l=1$ to $l=5$ from right to left) are described in equation (\ref{eq:SNRthre})
    }
    \label{fig:rician-errorall}
\end{figure} 
\begin{figure}
    \centering
    \includegraphics[width = 0.48\textwidth]{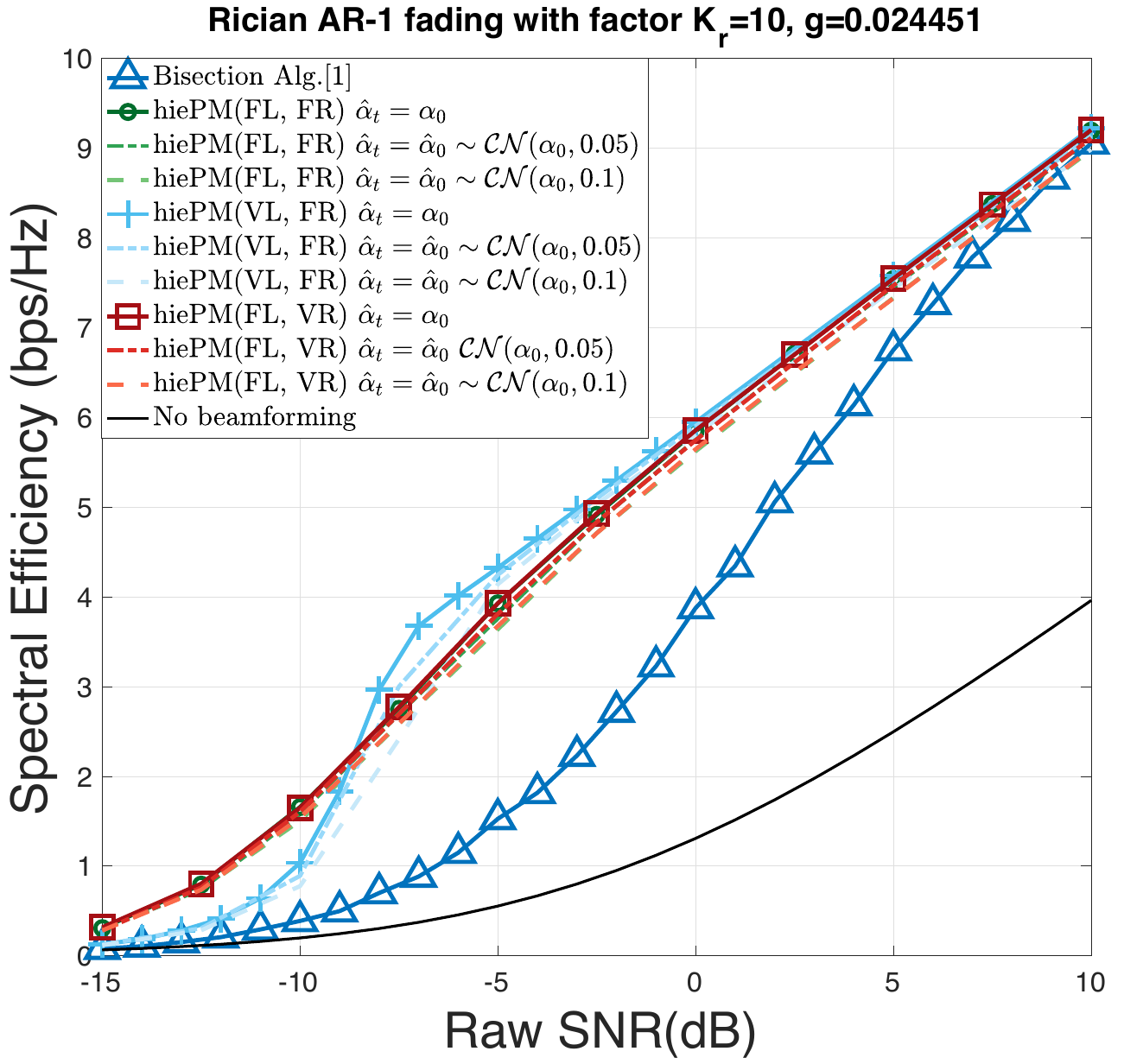}
    \caption{Comparison of the transmission rates obtained by \textit{hiePM} and the vanilla bisection algorithm of \cite{Alkhateeb2014} as a function of raw SNR  $P/\sigma^2$ under Rician AR-1 fading with factor $K_r =10$, and $g=0.024451$, (i.e. $T_c=2$). Initial access time $\tau=28$, achieved under 2 ms using the 5G NR PRACH format B4 \cite{Lin20185GNR} ($E[\tau]=28$ for variable-length stopping). The transmission rate is given by (8) - in original manuscript  with the final beamforming vector $\hat{\textbf{w}}$ designed by the respective algorithm.  }
    \label{fig:rician_rate}
\end{figure}
In this subsection, we will discuss the channel coherence time and how our initial beam alignment algorithm works in a time-varying channel scenario. We verify our framework by extending our algorithms to be adapted to a simple Rician AR-1 model. Let us consider a Rician AR-1 fading channel of factor $K_r$ with perfect knowledge of the operating SNR (large-scale fading) as well as perfect frequency/phase synchronization, i.e. the fading coefficient is given as
\begin{equation} \label{Rician}
\alpha_t = \sqrt{\frac{K_r}{1+K_r}} \mu + \sqrt{\frac{1}{1+K_r}} \beta_t, \ t=0,1,2,...,\tau,
\end{equation}
where $\mu=1$ and $\beta_t \sim CN(0,1)$ is the complex Gaussian diffusion AR process given as
\begin{equation} \label{AR-d}
\beta_{t+1} = \beta_t \sqrt{1-g}  + e_t\sqrt{g}, \ \ t=0,1,2,...,\tau,
\end{equation}
where $e_t\sim \mathcal{CN}(0,1)$ is the independent noise term. The correlation parameter $g$ is set such that 
\begin{equation}
 1-(1-g)^{14 T_c} = 0.5,
\end{equation}
where $T_c$ is the $50\%$ time of the diffusion $\beta_t$ in ms (recall that we assume a system with 14 beam slots in $1$ ms). Combining (\ref{Rician}) and (\ref{AR-d}), the Rician AR-1 model can be written as 
\begin{equation}
\begin{aligned}
\label{AR-1}
    \alpha_{t+1} = &\sqrt{\frac{K_r}{1+K_r}} \mu  + \Bigg( \alpha_t-\sqrt{\frac{K_r}{1+K_r}} \mu \Bigg) \sqrt{1-g} \\
    &+ e_t\sqrt{\frac{g}{1+K_r}}, \ \ t=0,1,2,...,\tau.
\end{aligned}
\end{equation}

Fig.~\ref{fig:rician-errorall} demonstrates the robustness of $hiePM$ to the Rician AR-1 fading channel model described above with coherence time $T_c=2$ ms of the AR process $\beta_t$, and a Rician factor $K_r=10$ (this is a reasonable value, e.g. indoor mmWave channel models \cite{Kfactor_MUKHERJEE}). We again use an erroneous/mismatched and fixed estimate of the fading coefficient $\hat{\alpha}_t = \hat{\alpha}_0 \sim \mathcal{CN}(\alpha_0,\sigma_{\alpha}^2)$ for $t = 1, 2, \ldots, \tau$. In particular, we compare the performance achieved by our $hiePM$ algorithms with different degrees of knowledge of the fading estimate (i.e. $\sigma_{\alpha}^2 = 0, 0.05,0.1$) against the performance obtained by the bisection algorithm of \cite{Alkhateeb2014}. As expected, the performance of the probability of error worsens for both algorithms in a time-varying fading, as compared to the static model $\alpha=1$ in Fig.~\ref{fig:sim-errorall}. However, even under a mismatched and fixed estimate of the fading coefficient $\hat{\alpha}_t$, our main conclusions still hold. In particular, the performance degradation in spectral efficiency due to the time-varying channel is almost negligible, as we show in Fig.~\ref{fig:rician_rate}. Note that, this is the effect of the time-varying channel during the initial access phase, whereas in the communication phase the spectral efficiency is calculated by (\ref{eq:data_rate}) because our focus is the impact of a time-varying channel on the initial beam alignment. Our variable resolution algorithm $hiePM$(FL,VR) (with opportunistic choice of final beamwidth) is unaffected in terms of spectral efficiency, while the $hiePM$(FL,FR) and $hiePM$(VL,FR) cases incur a small loss of spectral efficiency due to the degree of the mismatched estimate (correlated to the severity of $\sigma_{\alpha}^2$).

\section{Conclusion and Future Work}
In this paper, we addressed the initial access problem for mmWave communication with beamforming techniques. With a single-path channel model, the proposed sequential beam search algorithm $hiePM$ demonstrates a systematic way of actively learning an optimal beamforming vector from the hierarchical beamforming codebook of \cite{Alkhateeb2014}.  

Using a single-path channel model, we characterize the performance of the proposed learning algorithm 
$hiePM$ by the resolution and the error probability of learning the AoA, which are closely related to the link quality established by the final beamforming vector. We analyze $hiePM$ by giving an upper bound on the expected search time $\tau_{\epsilon,\delta}$ required to achieve a resolution $\frac{1}{\delta}$ and error probability $\epsilon$ in Theorem~\ref{thm}. As a corollary, we provide an upper bound on the error probability achieved with a search time $\mathbb{E}[\tau_{\epsilon,\delta}]$, and resolution $\frac{1}{\delta}$ for $hiePM$ in Corollary~\ref{cor_eps}. We also specialize our analysis and compare the error exponent obtained by $hiePM$ and the bisection algorithm of \cite{Alkhateeb2014}. A higher error exponent is shown across a wide range of raw SNR even when only 1-bit of information about the measurement is available to $hiePM$. The numerical simulations show a significant improvement on the spectral efficiency over the previous vanilla bisection algorithm of \cite{Alkhateeb2014} and the random search algorithm modelled as a best case of \cite{Abari2016,Song19}, demonstrating a first work of possible standalone mmWave communication.

Future directions of this work include generalizing the channel model and considering multiple paths, as well as learning the fading coefficient together with the direction during beam search. On the theoretical end, closing the gap between the upper bound of error probability and its actual performance (demonstrated in Fig.~\ref{fig:sim-errorall}) is worth pursuing for theoretical interest. On the practical side, reducing the computation complexity of the posterior calculations and required statistics will be helpful for implementation purposes. 

\bibliographystyle{IEEEbib}
\bibliography{./refs}

\appendices

\section{Optimal Threshold for 1-bit measurement model}
The complete 1-bit measurement model in Sec. \ref{sec:posterior_update} is written as
\begin{equation} 
\begin{aligned}
    y_{t+1} &= \alpha \sqrt{P} \*w_{t+1}^H \*a(\phi) + \*w_{t+1}^H \*n_{t+1} \\
    z_{t+1} &= \mathds{1}(|y_{t+1}|^2 > v_t) \\
            &= \mathds{1}(\phi\in D_{l_t}^{k_t} ) \oplus u_{t}(\phi), \ \ u_{t}(\phi) \sim \text{Bern}(p_t(\phi)).
\end{aligned}
\end{equation}

\begin{lemma} \label{lemma:opt_threshold}
The threshold $v_t$ that minimizes the maximum flipping probability $p_t(\phi)$ for all $\phi$ is given by the solution of the following equation
\begin{equation} \label{eq:threshold}
    \int_0^{v_t} \text{Rice}(x; PG ,\sigma^2)\ dx = \int_{v_t}^{\infty} \text{Rice}(x; Pg ,\sigma^2)\ dx,
\end{equation}
where 
\begin{equation} \label{eq:gG}
    \begin{aligned}
        G&:=\min_{\phi\in D_{l_t}^{k_t}} |\*w^H(D_{l_t}^{k_t}) \*a(\phi)|^2 \\
        g&:=\max_{\phi\in [\theta_{\text{min}},\theta_{\text{max}}] \setminus D_{l_t}^{k_t}} |\*w^H(D_{l_t}^{k_t}) \*a(\phi)|^2.
    \end{aligned}
\end{equation}
\end{lemma}
\begin{proof}
Since $|y_t|\sim \text{Rice} (P|\*w_t^H \*a(\phi)|,\sigma^2)$, we can write the flipping probability $p_t(\phi)$ as:\\
if $\phi\in D_{l_t}^{k_t}$,
\begin{equation} \label{eq:vt_1}
\begin{aligned}
    p_t(\phi) &= \int_{0}^{v_t} \text{Rice}(x;P|\*w^H(D_{l_t}^{k_t}) \*a(\phi)|^2, \sigma^2)\ dx\\
              &\leq \int_{0}^{v_t} \text{Rice}(x;PG, \sigma^2)\ dx,
\end{aligned}
\end{equation}
and if $\phi\notin D_{l_t}^{k_t}$,
\begin{equation} \label{eq:vt_2}
\begin{aligned}
    p_t(\phi) &= \int_{v_t}^{\infty} \text{Rice}(x;P|\*w^H(D_{l_t}^{k_t}) \*a(\phi)|^2, \sigma^2)\ dx\\
              &\leq \int_{v_t}^{\infty} \text{Rice}(x;Pg, \sigma^2)\ dx,
\end{aligned}
\end{equation}
where the upper bound in (\ref{eq:vt_1}) and (\ref{eq:vt_2}) is reached by the minimizer and maximizer in (\ref{eq:gG}), respectively. Since (\ref{eq:vt_1}) is increasing in $v_t$ and (\ref{eq:vt_2}) is decreasing in $v_t$, setting them equal gives the minimax optimizer.
\end{proof}

\section{Average Log-Likelihood and the Extrinsic Jensen-Shannon Divergence} \label{sec:analysis-prelim}
 \label{reEJS}
Our analysis follows similarly to \cite{Chiu2018,Chiu2016} which analyzed a feedback coding scheme and an abstract adaptive target search algorithm, respectively. The analysis is based on the behavior of the posterior of the AoA
\begin{equation}
\pi_i(t) := \prob*{ \phi = \theta_i |y_{1:t}, \*w_{1:t}} \quad \text{for } i=1,2,...,\frac{1}{\delta},
\end{equation}
over time $t=1,2,...$. Recall that the stopping time of $hiePM$ is given by
\begin{equation}
	\tau_{\epsilon,\delta} = \min \{t:  1-  \max_i \pi_{i}(t) \leq \epsilon \},
\end{equation}
which is the first hitting time of the posterior to the $\epsilon$ corner of the probability simplex in $\mathbb{R}^{\frac{1}{\delta}}$.

Convex functionals on the probability simplex such as entropy, KL divergence, etc. have been shown to be useful for analyzing adaptive systems \cite{Chiu2016,Shayevitz2016,Naghshvar2015,Li2014} with the Bayes' rule dynamics on the posterior distribution. Particularly, the functional average log-likelihood \cite{Naghshvar2015} has shown its usefulness in analyzing the behaviour of the posterior in feedback coding systems \cite{Chiu2018}, dynamic spectrum sensing \cite{Lalitha2017}, hypothesis testing \cite{Naghshvar2013ASH}, active learning \cite{Naghshvar2013ASH}, etc. Here, we review some useful concepts through the context of AoA estimation with sequential beamforming:

The average log-likelihood of the posterior $\.\pi(t)$ is defined as
\begin{equation}
    U(t)  := \sum_{i=1}^{1\slash \delta} \pi_{i}(t)  \log \frac{\pi_i(t)}{1-\pi_i(t)}.
\end{equation}
For any beamfomring strategy $\gamma: \.\pi(t) \rightarrow \*w_{t+1}$, $U(t)$ has the following useful properties  
\begin{enumerate}
    \item  $U(t)$ is a submartingale $w.r.t.$ $\.\pi(t)$ with expected drift $EJS$:
    \begin{equation}
        \expect{ U(t+1) | \.\pi(t) } =  U(t) + EJS(\.\pi(t), \gamma),
    \end{equation}
    where $EJS$ is the Extrinsic Jensen-Shannon divergence, defined as
    \begin{equation} %\label{eq:EJS}
        EJS(\.\pi(t), \gamma ) = \sum_{i=1}^{1\slash \delta} \pi_i(t) D\[( P_{y_{t+1}|  i , \gamma} \Big\| P_{y_{t+1}| \neq i, \gamma } \]),
    \end{equation}
    with 
    \begin{equation}
        P_{y_{t+1} | \theta_i , \gamma } := f(y_{t+1} |\phi = \theta_i, \*w_{t+1} = \gamma(\.\pi(t)) ), \\
    \end{equation} and 
    \begin{equation}
        \begin{aligned}
            P_{y_{t+1}| \neq i, \gamma } 
            &= \sum_{ j \neq i} \frac{\pi_{j}(t)}{1-\pi_i(t)}  P_{y_{t+1}|  j , \gamma }.
        \end{aligned}
    \end{equation}
    \item The initial value $U(0) = -\log (\frac{1}{\delta}-1)$ is related to the resolutuion $\frac{1}{\delta}$
    \item Level crossing of $U$ is related to the error probability as
    $
        \pi_i(t)<1-\epsilon\  \forall\  i  \ \Rightarrow\  U(t)< \log \frac{1- \epsilon}{\epsilon}.
    $
\end{enumerate}

Analyzing the random drift of $U(t)$ from time 0 with the initial value $U(0)$ up to the first crossing time $\nu:= \min\{ t: U(t) \geq \log \frac{1}{\epsilon} \}$ is closely related to the expected drift given by $EJS$. In particular, we can then establish an upper bound on the expected stopping time $\expect{\tau_{\epsilon,\delta}}$ in terms of the predefined outage probability $\epsilon$ and resolution $1\slash \delta$.  Specifically, we have the following theorem:
\begin{fact}[Theorem 1 in~\cite{Naghshvar2015}] \label{fact:EJSthm}
Define 
\begin{equation}
    \tilde{\pi} := 1- \frac{1}{1+\max \{ n, \log (1\slash \epsilon) \}}.
\end{equation}
For any feedback coding scheme $\gamma$, if \begin{equation}
    EJS(\.\pi(t), \gamma) \geq R \quad \forall t \geq 0
\end{equation}
and 
\begin{equation}
    EJS(\.\pi(t), \gamma) \geq \tilde{\pi} E \quad \forall  t \geq 0 \text{ s.t. } \max_{i} \pi_i(t) \geq \tilde{\pi},
\end{equation}
the expected decoding time associated with error probability $\epsilon$ is bounded by
\begin{equation}
    E[\tau_{\delta,\epsilon} ] \leq   \frac{\log (1\slash \delta)}{R} + \frac{\log (1\slash \epsilon)}{E}  + o(\log (\frac{1}{ \delta \epsilon} )),
\end{equation}
where $o(\cdot)$ is such that $\frac{o(\log (\frac{1}{ \delta \epsilon}))}{\log (\frac{1}{ \delta \epsilon})} \rightarrow 0$ as $\delta \rightarrow 0$ or $\epsilon \rightarrow 0$.
\end{fact}

On the other hand, in order to give a lower bound for the EJS divergence, it is useful to introduce the Jensen Shannon (JS) divergence \cite{Lin1991}, defined as
\begin{equation}
JS(\.\pi(t), \gamma) := \sum_{i=1}^{1\slash \delta} \pi_i(t) D\[( P_{y_{t+1}|  i , \gamma} \Big\| P_{y_{t+1}} \]),
\end{equation}
where $P_{y_{t+1}} = f(y_{t+1}) = \sum_{i=1}^{1\slash\delta}  P_{y_{t+1}|  i , \gamma }.$

\begin{fact}[Lemma 2 in \cite{Naghshvar2015}] \label{fact:JS_EJS}
The EJS divergence is lower bounded by the Jensen Shannon (JS) divergence : 
\begin{equation}
    EJS(\.\pi(t), \gamma) \geq JS(\.\pi(t), \gamma).
\end{equation}
\end{fact}

\section{Variable-length analysis of hierarchical Posterior Matching}   \label{sec:analysis-proof}

Here we provide the variable-length analysis of hierarchical Posterior Matching by using the EJS. Throughout this section, we will focus on the settings and assumptions in Theorem \ref{thm}, where the beam pattern is perfect (Assumption \ref{asm:idealbeam}) and 1-bit measurements are used. The corresponding EJS is written as 
\begin{equation} %\label{eq:EJS}
        EJS(\.\pi(t), \gamma_h ) = \sum_{i=1}^{1\slash \delta} \pi_i(t) D\[( P_{\hat{y}_{t+1}|  i , \gamma_h} \Big\| P_{\hat{y}_{t+1}| \neq i, \gamma_h } \])
\end{equation}
    with the 1-bit measurement model
\begin{equation}
    P_{\hat{y}_{t+1} | i , \gamma } := \text{Bern}(\hat{y}_{t+1} \oplus \mathds{1}(\theta_i\in D_{l_t}^{k_t} ); p[l_t] ).
\end{equation}
\subsection{Proof of Theorem \ref{thm}}
Let $\gamma_h$ be the $hiePM$ feedback coding scheme. By the same method in \cite{Chiu2018}, the EJS can be lower bounded by
\begin{align}
    EJS(\.\pi(t),\gamma_h) &\geq I( 1\slash 3 ;p[l_{t+1}]) , \ \forall \ t  \label{eq:R_lb_thm} \\
    EJS(\.\pi(t),\gamma_h) &\geq \tilde{\.\pi} C_1(p[\log_2 (1\slash\delta)  ]) ,\ \forall \max_i \.\pi_i \geq  \tilde{\.\pi} \label{eq:E_lb_thm}
\end{align}
(for completeness, we include the proof of equation (\ref{eq:R_lb_thm}) and (\ref{eq:E_lb_thm}) in Lemma \ref{lemma:EJSori}). By (\ref{eq:R_lb_thm}), (\ref{eq:E_lb_thm}) and Fact \ref{fact:EJSthm}, we immediately have
\begin{equation}
\begin{aligned}
    \expect{\tau_{\epsilon,\delta}} \leq  \frac{\log (1\slash \delta)}{I(1\slash 3;p[1])} + \frac{\log (1\slash \epsilon)}{C_1(p[\log_2 (1\slash\delta)])}   + o( \log(\frac{1}{\delta \epsilon})).
\end{aligned}
\end{equation}
The gap from $I(1\slash 3;p[1])$ to $I(1\slash 3;p[l])$ is done similarly to \cite{Chiu2016}: we need to further show that $hiePM$ is able to zoom-in to higher level beamforming and effectively obtain less noisy measurements ($p\big[\frac{K_0 \lceil \log\log \frac{1}{\delta } \rceil}{\log2} -1 \big]< p[1]$) for most of the time during initial access (Lemma \ref{lemma:EJS_lq}). Indeed, let $E_t = \{ l_{t+1} < \frac{K_0 \lceil \log\log \frac{1}{\delta } \rceil}{\log2} -1 \}$ be the event of using a lower level codeword, and let $F_n = \bigcup_{t=n}^{\infty} E_t$, by the total expectation theorem and the union bound we have
\begin{equation} \label{n_bound_1}
    \begin{aligned}
       \expect{\tau_{\epsilon,\delta}} &= \int_{\Omega} \tau_{\epsilon,\delta} \ d\mathbb{P} \leq \sum_{t=n}^{\infty} \int_{E_t} \tau_{\epsilon,\delta} \ d\mathbb{P} + \int_{F_n^C} \tau_{\epsilon,\delta} \ d\mathbb{P} \\
       &= \sum_{t=n}^{\infty} \int_{E_t} \expect{ \tau_{\epsilon,\delta} | \.\pi(t) } \ d\mathbb{P} + \int_{F_n^C} \expect{ \tau_{\epsilon,\delta} | \.\pi(n) }\ d\mathbb{P}.\\
    \end{aligned}
\end{equation}
% In fact, we show (in Lemma \ref{lemma:lq}) that the probability of using lower level beamforming $l_t < l$ decays exponentially fast in time, i.e.
%\begin{equation} \label{eq:lq}
%    \prob{l_t < l} \leq k_0 e^{- E_0 t}.
%\end{equation}

By the time homogeneity of the Markov Chain $\.\pi(t)$ together with Lemma \ref{lemma:EJSori}, we can upper bound the two terms associated with the ``good'' event $F_n^C$ and the ``bad'' but low probability event $E_t$ in (\ref{n_bound_1}) as
\begin{equation} \label{eq:Et}
\begin{aligned}
    \int_{E_t} &\expect{ \tau_{\epsilon,\delta} | \.\pi(t) } \ d\mathbb{P} \leq \prob{E_t}  \times \\
    & \[( t+  \frac{\log \frac{1}{\delta} }{I(1\slash 3;p[1])} + \frac{\log \frac{1}{\epsilon}}{C_1(p[\log_2 (1\slash\delta)])} +  o( \log(\frac{1}{\delta \epsilon}))  \]),
\end{aligned}
\end{equation}
and
\begin{equation} \label{eq:Fnc}
\begin{aligned}
    \int_{F_n^C} &\expect{ \tau_{\epsilon,\delta} | \.\pi(n) } \ d\mathbb{P} \leq  \\
    &  n+  \frac{\log \frac{1}{\delta} }{I(1\slash 3;p[l'])} + \frac{\log \frac{1}{\epsilon}}{C_1(p[\log_2 (1\slash\delta)])} +  o( \log(\frac{1}{\delta \epsilon})),
\end{aligned}
\end{equation}
where $l' := \frac{K_0 \lceil \log\log \frac{1}{\delta } \rceil}{\log2} -1 $. Plugging (\ref{eq:Et}) and (\ref{eq:Fnc}) back to (\ref{n_bound_1}) and further with Lemma~\ref{lemma:lq} upper bounding $\prob{E_t}$, we have
\begin{equation} \small
    \begin{aligned} \label{n_bound_2}
        \expect{\tau_{\epsilon,\delta}} &\leq  \frac{k_0 e^{-n E_0}}{1-e^{-E_0}}  \[( n + \frac{\log \frac{1}{\delta} }{I(1\slash 3;p[1])} + \frac{\log \frac{1}{\epsilon}}{C_1(p[\log_2 (1\slash\delta)])}  \])  \\
        & + n+  \frac{\log \frac{1}{\delta} }{I(1\slash 3;p[l'])} + \frac{\log \frac{1}{\epsilon}}{C_1(p[\log_2 (1\slash\delta)])} +  o( \log(\frac{1}{\delta \epsilon}))
    \end{aligned}
\end{equation}
for $n>\frac{(l'+1)\log 2}{K_0}$. Finally, letting $n= \lceil \log\log \frac{1}{\delta \epsilon} \rceil$ in equation (\ref{n_bound_2}) we conclude the assertion of the theorem.

\subsection{Technical Lemmas} \label{sec:lemmas}

\begin{lemma} \label{lemma:EJSori}
Using the $hiePM$ beamforming strategy $\gamma_h$ on codebook $\mathcal{W}^S$  with $S = \log_2 (1\slash\delta)$, we have
    \begin{align}
        EJS(\.\pi(t),\gamma_h) &\geq I( 1\slash 3 ; p[l_{t+1}]) , \ \forall \ t  \label{eq:R_lb} \\
        EJS(\.\pi(t),\gamma_h) &\geq \tilde{\.\pi} C_1(p[\log_2 (1\slash\delta)  ]) ,\ \forall \max_i \.\pi_i \geq  \tilde{\.\pi} \label{eq:E_lb}
    \end{align}
\end{lemma}

\begin{proof}
%\begin{comment}
We first prove equation (\ref{eq:E_lb}). By the selection rule of $hiePM$, the last level beamforming $\*w_t = \*w(D_S^{k_t})$ is used whenever $\max_{i} \pi_i(t) \geq \tilde{\pi} > 1\slash 2$. Therefore, 
\begin{equation}
\begin{aligned}
     EJS(\.\pi(t), \gamma_h) &= \sum_{i=1}^{1\slash \delta} \pi_i(t) D\[( P_{\hat{y}_{t+1}|  i , \gamma_h} \Big\| P_{\hat{y}_{t+1}| \neq i, \gamma_h } \])\\
     &\geq \tilde{\pi} D\[( P_{\hat{y}_{t+1}| i , \gamma_h} \Big\| P_{\hat{y}_{t+1}| \neq i, \gamma_h} \]) \\
     &= \tilde{\pi}  D( \text{Bern}(1-p[S]) \| \text{Bern}(p[S]))  \\
     &= \tilde{\pi} C_1(p[\log_2 (1\slash\delta)  ]).
\end{aligned}
\end{equation}

It remains to show equation (\ref{eq:R_lb}). For notational simplicity, let 
\begin{equation}
    \rho \equiv \pi_{D_{l_{t+1}}^{k_{t+1}}}(t) := \sum_{\theta_i\in D_{l_{t+1}}^{k_{t+1}}} \pi_i(t)
\end{equation}
and $B^0 \equiv \text{Bern}(p[l_{t+1}])$, $B^1 \equiv \text{Bern}(1-p[l_{t+1}])$. We separate the proof into two cases:

If $\rho >  2 \slash 3$, we know that $l_{t+1} = S$ by the selection rule of $hiePM$. Therefore, the set $ D_{l_{t+1}}^{k_{t+1}}$ contains only 1 angle. Let $ D_{l_{t+1}}^{k_{t+1}} = \{\theta^*\} $, we have
\begin{equation}
\begin{aligned}
     &EJS(\.\pi(t), \gamma_h) =\sum_{i=1}^{1\slash \delta} \pi_i(t) D\[( P_{\hat{y}_{t+1}|  i , \gamma_h} \Big\| P_{\hat{y}_{t+1}| \neq i, \gamma_h } \])\\
    & =   \rho D\big( B^1 \big\| B^0 \big)   \\ 
     & + \sum_{i: \theta_i \neq \theta^*}  \pi_{i}(t)  D\[( B^0 \Big\| \frac{ \rho }{1-\pi_{i}(t)}   B^1 + \frac{1-\rho-\pi_i(t)}{1-\pi_i(t)} B^0 \]) \\[2mm]
     &\stackrel{(a)}{\geq}  D\[( B^0 \Big\| \frac{1}{2}  B^1 + \frac{1}{2} B^0 \])  \\
     & =  I(1\slash 2; p[l_{t+1}] ) \geq I(1\slash 3; p[l_{t+1}] ),  \\ 
\end{aligned}
\end{equation}
where (a) is by the fact that $ D( B^1 \| B^0 )  =  D( B^0 \| B^1 )   $ and that 
$D(  B^0  \| \alpha  B^1 + (1-\alpha) B^0 )$ is increasing in $\alpha$ for $0\leq \alpha \leq 1$, together with $\frac{ \rho }{1-\pi_{i}(t)} >  2\slash 3> 1\slash 2$.

For the other case where $\rho \leq 2 \slash 3$, again by the selection rule of $hiePM$, we have $1 \slash 3 \leq \rho \leq 2 \slash 3$. Now we can lower bound the $EJS$ as
\begin{equation}
\begin{aligned}
     EJS(\.\pi(t), \gamma_h) &\stackrel{(a)}{\geq} JS(\.\pi(t), \gamma_h)\\
     %&=   \hspace{-2.5mm} \sum_{i:\theta_i \in D_{l_{t+1}}^{k_{t+1}} }  \hspace{-2.5mm} \pi_{i}(t)  D \[( B^0  \Big\|  \rho B^0 + (1-\rho) B^1   \]) \\
     %& \quad + \hspace{-2.5mm} \sum_{i:\theta_i \notin D_{l_{t+1}}^{k_{t+1}} }  \hspace{-2.5mm} \pi_{i}(t)  D \[( B^1  \Big\|  \rho B^0 + (1-\rho) B^1   \])\\
     & = \rho D \[( B^1  \Big\|  \rho B^0 + (1-\rho) B^1  \]) \\
     & \qquad + (1-\rho) D \[( B^0  \Big\|  \rho B^0 + (1-\rho) B^1  \])   \\
     &= I(\rho; p[l_{t+1}] ) \stackrel{(b)}{\geq} I( 1\slash 3 ; p[l_{t+1}] )
\end{aligned}
\end{equation}
where (a) is by Fact \ref{fact:JS_EJS} and (b) is by the concavity of the mutual information w.r.t the input distribution, the symmetric of $I(\rho; p[l_{t+1}] )$ around $\rho = 1\slash 2$ for symmetric channels, and together with $1\slash 3 \leq \rho \leq 2\slash 3$. This concludes the assertion.
%\end{comment}
\end{proof}

\begin{lemma} \label{lemma:lq}
Using the $hiePM$ beamforming strategy $\gamma_h$ on codebook $\mathcal{W}^S$  with $S = \log_2 (1\slash\delta)$, we have 
\begin{equation}
    \prob{E_t} :=\prob{l_{t+1} \leq l} \leq k_0 e^{- E_0 t} 
\end{equation}
for all $ t>\frac{\log (l+1)\log 2 }{K_0}$, where $k_0=e^{\frac{K_0(l+1)\log 2}{(2l \log 2 + K_0)^2}}$, $E_0 = \frac{K_0^2}{2(2l \log 2 + K_0)^2}$, and $K_0>0$ is a constant defined in Lemma \ref{lemma:EJS_lq}.
\end{lemma}
\begin{proof}
Let $\.\pi^{\{l\}}(t)$ and $U^{\{l \}}(t)$ be the nested posterior of level $l$ and its averaged log-likelihood, defined in Appendix \ref{sec:nested_posterior}, equations (\ref{eq:nestedpi}) and (\ref{eq:nestedU}). Note that $U^{\{l \}}(t) \geq \log 2$ implies that $\max_q \pi^{\{l\}}_q(t) \geq 2\slash 3$, and in turn implies that $l_{t+1} > l$ by the selection rule of $hiePM$. Therefore, it suffices to show that 
\begin{equation} \label{eq:transU}
    \prob{U^{\{l \}}(t) < \log 2} \leq k_0 e^{- E_0 t} \quad \forall t>T_0.
\end{equation}
We will show (\ref{eq:transU}) using submartingale properties of $U^{\{l \}}(t)$ with Azuma's inequality \cite{Chung2006}. Indeed, by Lemma~\ref{lemma:EJS_lq} in Appendix~\ref{sec:nested_posterior}, $U^{\{l \}}(t) - K_0 t$ is a submartingale. Furthermore, we have bounded differences for this submartingale, $i.e.$
\begin{equation}
    |U^{\{l \}}(t+1) - U^{\{l \}}(t) + K_0  | \leq  2l \log 2 + K_0 ,
\end{equation} for all $t\geq 0$. Hence we have
\begin{equation}
\begin{aligned}
    &\mathbb{P}( U^{\{l \}}(t) < \log 2 ) \\
    &= \mathbb{P} \big( U^{\{l \}}(t) - K_0  t - U^{\{l \}}(0) < ( l +1) \log 2 - K_0 t \big) \\
    & \overset{(a)}{\leq} \exp\[( -  \frac{ (  ( l+1) \log 2 - K_0 t  )^2}{ 2t (2l \log 2 + K_0)^2  }  \]) \overset{(b)}{\leq} k_0 e^{-E_0 t} \\
    %& = \exp\[(- \frac{K_0^2 t}{ 2 B^2 }   \]) \exp \[( \frac{K_0  ( l +1) \log 2 }{ B^2 } \]) \exp \[( - \frac{(  ( l +1) \log 2 )^2}{ 2t B^2 } \]) \\
    %& \leq \exp\[( \frac{K_0  ( l +1) \log 2 }{ B^2 } \])  \exp\[( - \frac{K_0^2 t}{ 2 B^2 }   \])  := k_f e^{-t E_0},
\end{aligned}
\end{equation}
for $t> \frac{\log (l+1)\log 2 }{K_0}$, where $(a)$ is by Azuma's inequality and $(b)$ is by expanding the quadratic terms and rearrangements.
\end{proof}

\section{Nested Posterior and its EJS} \label{sec:nested_posterior}

In this section, we introduce the nested posterior and its EJS lower bound (Lemma \ref{lemma:EJS_lq}), which are used for proving Lemma \ref{lemma:lq}. Let posterior  $\pi_i(t)$ with $i=1,2,...,2^S$. We define a nested posterior $\.\pi^{\{l\}}$ of level $l<S$ with length $2^l$ as 
\begin{equation} \label{eq:nestedpi}
    \pi_q^{\{l\}} (t) := \sum_{i\in \text{bin}(q)} \pi_i(t), \quad q = 1,2,...,2^l,
\end{equation}
where $\text{bin}(q) := \{(q-1)2^{S-l}+1,...,q2^{S-l}\}$. Further, we define the functional log-likelihood on $\.\pi^{\{l\}}$ as
\begin{equation} \label{eq:nestedU}
    U^{\{l\}}(t)  := \sum_{q=1}^{2^l}\pi_q^{\{l\}}(t)  \log \frac{\pi_q^{\{l\}}(t)}{1-\pi_q^{\{l\}}(t)}.
\end{equation}
For any beamforming strategy $\gamma: \.\pi^{\{S\}}(t) \rightarrow \*w_{t+1}$ on the level $S$ posterior, the level $l<S$ log-likelihood $U^{\{l\}}(t)$ is a submartingale $w.r.t.$ $\.\pi(t)$. The expected drift can be written as
\begin{equation}
\begin{aligned}
    &\expect{ U^{\{l\}}(t+1) | \.\pi(t) } \\
    &\qquad =  U^{\{l\}}(t) + EJS(\.\pi^{\{l\}}(t), \gamma ; \.\pi(t)),
\end{aligned}
\end{equation}
where
\begin{equation} %\label{eq:EJSnest}
\begin{aligned}
     &EJS(\.\pi^{\{l\}}(t), \gamma ; \.\pi(t)) \\
     & \qquad= \sum_{q=1}^{2^l} \pi^{\{l\}}_q(t) D\[( P_{y_{t+1}|  q , \gamma} \Big\| P_{y_{t+1}| \neq q, \gamma } \])
\end{aligned}
\end{equation}
with 
\begin{equation}
\begin{aligned}
    &P_{y_{t+1} | \in \text{bin}(q) , \gamma } :=  \frac{1}{\sum_{i\in \text{bin}(q)} \pi_i(t) } \\
    & \quad  \times \sum_{i\in \text{bin}(q)} \pi_i(t) f\big(y_{t+1} \big|\phi =\theta_i , \*w_{t+1} = \gamma(\.\pi(t)) \big) 
\end{aligned}
\end{equation} and 
\begin{equation}
\begin{aligned}
    &P_{y_{t+1} | \notin \text{bin}(q) , \gamma } :=  \frac{1}{\sum_{i\notin \text{bin}(q)} \pi_i(t) } \\
    & \quad  \times \sum_{i\notin \text{bin}(q)} \pi_i(t) f\big(y_{t+1} \big|\phi =\theta_i , \*w_{t+1} = \gamma(\.\pi(t)) \big).
\end{aligned}
\end{equation}

\begin{lemma} \label{lemma:EJS_lq}
With same assumptions as Theorem \ref{thm}, using $hiePM$ with codebook $\mathcal{W}^S$ on $\.\pi(t) \equiv \.\pi^{\{S\}}(t)$, we have 
\begin{equation}
\begin{aligned}
    & EJS(\.\pi^{\{l\}}(t),\gamma_h; \.\pi(t)) \geq K_0 := \min \Big\{  I \Big(\frac{1}{3},p[1] \Big), \\
            &\frac{2}{3} D\Big(\frac{1}{3} \text{Bern}(1-p[1]) + \frac{2}{3} \text{Bern}(p[1])  \Big\| \text{Bern}(p[1])  \Big) \Big\} \\
\end{aligned}
\end{equation}
for all $t>0$, for any $l<S$.
\end{lemma}
\begin{proof}
%For notational simplicity, define $\rho := \.\pi^{\{l\}}_{q_t}(t)$. 
Given any $l<S$, if the selected codeword $\*w(D_{l_{t+1}}^{k_{t+1}})$ is such that $l_{t+1}\leq l$, by Lemma \ref{lemma:EJSori} we conclude the results. If otherwise $l_{t+1}>l$, then for any $\theta_i\in D_{l_{t+1}}^{k_{t+1}}$, $i\in \text{bin}(q_t)$ for some $q_t$.  For notational simplicity, let 
%\begin{equation}
    $\rho \equiv \pi_{D_{l_{t+1}}^{k_{t+1}}}(t) := \sum_{\theta_i\in D_{l_{t+1}}^{k_{t+1}}} \pi_i(t)$
%\end{equation}
and $B^0 \equiv \text{Bern}(p[l_{t+1}])$, $B^1 \equiv \text{Bern}(1-p[l_{t+1}])$. We have
\begin{equation} %\label{eq:EJSnest}
\begin{aligned}
     &EJS(\.\pi^{\{l\}}(t), \gamma ; \.\pi(t)) \\
     &= \sum_{q=1}^{2^l} \pi^{\{l\}}_q(t) D\[( P_{\hat{y}_{t+1}|  q , \gamma} \Big\| P_{\hat{y}_{t+1}| \neq q, \gamma } \]) \\
     &\stackrel{(a)}{\geq} \frac{2}{3} D( \rho B^1 + (1-\rho) B^0 \| B^0 ) \stackrel{(b)}{\geq} \frac{2}{3} D( \frac{1}{3} B^1 + \frac{2}{3} B^0 \| B^0 ) \\
     &\geq \frac{2}{3} D\Big(\frac{1}{3} \text{Bern}(1-p[1]) + \frac{2}{3} \text{Bern}(p[1])  \Big\| \text{Bern}(p[1])  \Big) \Big\}.
\end{aligned}
\end{equation}
where (a) and (b) are by the selection rule of $hiePM$ that $\.\pi^{\{l\}}_{q_t}(t)>2\slash 3$ whenever $l_t > l$ and that $1 \slash 3 \leq \rho \leq 2 \slash 3$. This concludes the assertion.
\end{proof}

\end{document}